\documentclass{article}
\usepackage[utf8]{inputenc}
\usepackage{amsmath,amsthm, amssymb, amsfonts}
\usepackage{fullpage}
\usepackage[pdfencoding=auto]{hyperref}
\usepackage{graphicx}
\usepackage{wasysym}
\usepackage{todonotes}
\usepackage{comment}
\usepackage{mathtools}
\usepackage{wrapfig}
\usepackage{microtype}
\usepackage[algo2e]{algorithm2e}
\usepackage{xspace}
\usepackage{svg}

\usepackage{algorithm,algorithmic}

\definecolor{darkgreen}{rgb}{0,0.5,0}
\definecolor{darkblue}{RGB}{37,70,160}
\definecolor{amber}{rgb}{1.0, 0.75, 0.0}
\definecolor{aquamarine}{rgb}{0.4, 0.8, 0.66}
\definecolor{purple}{RGB}{186, 121, 246}
\definecolor{green}{rgb}{0.2,0.6,0.15}

\usepackage{tikz,xcolor}
\usepackage{tikz-3dplot}
\usetikzlibrary{3d,decorations.text,shapes.arrows,positioning,fit,backgrounds}
\usepackage{pgflibraryplotmarks}
\usepackage{pgflibrarysnakes}
\usetikzlibrary{calc,fadings,snakes,matrix,positioning,spy,decorations.pathreplacing,decorations.text}

\tikzset{pics/fake box/.style args={
#1 with dimensions #2 and #3 and #4}{code={\draw[gray,ultra thin,fill=#1]  (0,0,0) coordinate(-front-bottom-left) to
++ (0,#3,0) coordinate(-front-top-right) --++
(#2,0,0) coordinate(-front-top-right) --++ (0,-#3,0)
coordinate(-front-bottom-right) -- cycle;
\draw[gray,ultra thin,fill=#1] (0,#3,0)  --++
 (0,0,#4) coordinate(-back-top-left) --++ (#2,0,0)
 coordinate(-back-top-right) --++ (0,0,-#4)  -- cycle;
\draw[gray,ultra thin,fill=#1!80!black] (#2,0,0) --++ (0,0,#4) coordinate(-back-bottom-right)
--++ (0,#3,0) --++ (0,0,-#4) -- cycle;
\path[gray,decorate,decoration={text effects along path,text={CONV}}] (#2/2,{2+(#3-2)/2},0) -- (#2/2,0,0);
}
}}

\tikzset{circle dotted/.style={dash pattern=on 0.05mm off 2mm,line cap=round}}

\allowdisplaybreaks[1]
\usepackage{sistyle} 

\newcommand{\PW}{\mathcal{PW}}

\newcommand{\Integer}{\mathbb{Z}}
\newcommand{\Natural}{\mathbb{N}}

\newcommand{\Real}{\mathbb{R}}

\DeclareMathOperator{\sinc}{sinc}
\DeclareMathOperator{\supp}{supp}
\DeclareMathOperator{\DWF}{DWF}
\DeclareMathOperator{\WF}{WF}

\let\emptyset\varnothing

\newtheorem{theorem}{Theorem}[section]
\newtheorem*{theorem*}{Theorem}
\newtheorem{remark}[theorem]{Remark}
\newtheorem{definition}[theorem]{Definition}

\newtheorem*{remark*}{Remark}
\newtheorem*{proposition*}{Proposition}

\numberwithin{equation}{section}

\DeclareMathOperator{\suppp}{supp \,}

\definecolor{darkcandyapplered}{rgb}{0.64, 0.0, 0.0}

\title{Extraction of digital wavefront sets using applied harmonic analysis and deep neural networks}
\author{Héctor Andrade-Loarca\footnotemark[1]\and Gitta Kutyniok\footnotemark[1]\and Ozan \"Oktem\footnotemark[2]\and Philipp Petersen\footnotemark[3]}

\date{}
\begin{document}
\maketitle

\footnotetext[1]{Institut f\"ur Mathematik, Technische Universit\"at Berlin, 10623 Berlin, Germany, \texttt{$\{$kutyniok,andrade$\}$@math.tu-berlin.de}}
\footnotetext[2]{Department of Mathematics, KTH - Royal Institute of Technology,
SE-100 44 Stockholm, Sweden, \texttt{ozan@kth.se}}
\footnotetext[3]{
Mathematical Institute, University of Oxford, OX2 6GG, Oxford, UK, \texttt{Philipp.Petersen@maths.ox.ac.uk}}

\begin{abstract}
Microlocal analysis provides deep insight into singularity structures and is often
crucial for solving inverse problems, predominately, in imaging sciences. Of particular
importance is the analysis of wavefront sets and the correct extraction of those. In
this paper, we introduce the first algorithmic approach to extract the wavefront set 
of images, which combines data-based and model-based methods. Based on a celebrated property of the shearlet transform to unravel information on the wavefront set, we extract the wavefront set of an image by first applying a discrete shearlet transform and then feeding local patches of this transform to a deep convolutional neural network trained on labeled data.
The resulting algorithm outperforms all competing algorithms in edge-orientation and ramp-orientation detection.
\end{abstract}

\smallskip


\section{Introduction}
Many scientific and industrial real-world applications require a precise understanding of how a model parameter, represented by a function, is transformed under some process that is described by an operator.
Such analysis quickly becomes very challenging, and one attempt at simplifying it is to treat the singular (non-smooth) and smooth parts of the function separately.
In fact, a significant portion of the useful information is often contained in the singular part. For images, this singular part corresponds to edges, ridges, or ramps in the image.

Microlocal analysis is a powerful mathematical theory that aims to precisely describe how the singular part of a function, or more generally a distribution, is transformed when acted upon by an operator. Since its introduction in the early 1970s by Sato \cite{Sato:1971aa} and Hörmander \cite{Hormander:1971aa}, it has proven itself useful in both pure and applied mathematical research.
The crucial underlying observation in microlocal analysis is that the information about the location of the singularities (singular support) needs to be complemented with specifying directional information of the singularities. This extra (``microlocal'') information is key in elucidating the propagation of singularities by a certain class of operators. This class includes Fourier integral operators such as most differential and pseudo-differential operators as well as many integral operators arising in integral geometry. Such operators are frequently encountered in analysis, scientific computing, and physical sciences \cite{Hormander:1971aa, candes2007fast}.

Microlocal analysis is also particularly useful in inverse problems, where the goal is to reliably recover a hidden model parameter (signal) from a noisy transformed version.
The goal here is to recover the wavefront set of the signal (image) given the noisy realization of a transformed version of the signal.
Such applications frequently arise when using imaging/sensing technologies where the transform is a pseudo-differential or Fourier integral operator \cite{krishnan2015microlocal}.

A prime example exhibiting the utility of microlocal analysis in inverse problems appears in the context of imaging applications. Here, the image is {represented by} a function describing the interior structure of the object under study. 
Recovering the image is, however often not possible, either because the transformation relating the image to data is not invertible or because data is incomplete.
On the other hand, full knowledge of the image is not always necessary.
As an example, if one is looking for a tumor, then its location and shape are often sufficient for decision making whereas the exact values of the tumor density may be ignored or identified with another more specialized technique.
The location and shape of the tumor can be readily determined from the singular part of the image.

For the above reason, there exist a plethora of applications of microlocal analysis to tomographic imaging. In these applications, the transformation relating the image to data is the ray transform which can be interpreted as a Fourier integral operator. Thereby, it is possible to explicitly describe the relation between the wavefront set of the function (image) and its transformed version (tomographic data).
Such relationships are referred to as \emph{(microlocal) canonical relations}.

The canonical relation can also be used to identify which singularities can be recovered from data without explicitly computing the inverse ray transform.
This was done in \cite{quinto1993singularities} for the case when the 2D/3D ray transform is restricted to parallel lines and in \cite{quinto2008local} for an analysis in the region-of-interest limited angle setting.
A related principle was derived in \cite{katsevich1995helical} for the 3D ray transform restricted to lines given by helical acquisition that is common in medical imaging.
Similar principles hold for transforms integrating along other types of curves, for example, ellipses with foci on the x-axis and geodesics \cite{uhlmann2012geodesic}.
Another observation is that recovering a signal from its ray transform is less ill-posed if one knows the wavefront set a priori. This was demonstrated in \cite{devison1983limited} where the severely ill-posed reconstruction problem in limited angle tomography becomes mildly ill-posed if the wavefront set of the solution is provided as prior information, see \cite{quinto2008local} for an application of this principle to cryo-electron tomography.

Additionally, the concept of the wavefront set is of significant interest in mathematical image processing outside of inverse problems. 
Directed singularities, in particular edges, ridges, or ramps play an essential role in image processing to the extent that edge detection is one of the principal problems of this field.
This is due to the fact, that edges in images present boundaries of objects and carry most of the information of the associated physical scene \cite{4767769, marr1980theory, BINFORD1981205, brady1982computational}.
Additionally, it has been argued in \cite{marr1980theory} that the human visual cortex performs multiple operations of image processing,
the first of which is rough sketching involving edge detection.

\subsection{Wavefront set extraction}
A necessary first step before applying techniques from microlocal analysis in real-world applications in the way described above is to extract the wavefront set of functions. By the definition of the wavefront set, this amounts to estimating the asymptotic behavior of the localized Fourier transform of the function in question. 

\paragraph{Extraction from finitely many samples.}
In applications, where only finitely many point samples of the underlying function are known, estimating the asymptotic behavior of the Fourier transform is usually not possible. Indeed, we show in Section~\ref{sec:noWaveFrontSetExtractor} that \emph{every method} that seeks to approximately extract the wavefront set of functions from a function class $\mathcal{C} \subset L^2$ will \emph{fail} on a dense subset of $\mathcal{C}$ as soon as $\mathcal{C}$ contains at least all $k$-times differentiable functions for an arbitrary $k \in \Natural$.

The problem of extracting the wavefront set of functions from a function class $\mathcal{C}$ certainly becomes easier the smaller $\mathcal{C}$ gets. Indeed, in applications, where, for example, $\mathcal{C}$ models a class of images, it is reasonable to assume that $\mathcal{C}$ is a very small subset of $L^2$ and does not contain all $k$-times differentiable functions for any $k \in \Natural$. Hence, in this situation, wavefront set extraction from point samples explicitly designed for the class of images could be feasible. In fact, if the set $\mathcal{C}$ is so small that every function in $\mathcal{C}$ is uniquely determined through its samples on the grid, then wavefront set detection could, in principle, be performed through a large data base, which could be learned from $\mathcal{C}$. 

{From the above}, it is clear that a {useful} wavefront set extractor needs to {closely adapted to the underlying function class $\mathcal{C}$. This strong adaptation to $\mathcal{C}$ will henceforth be referred to as} our \emph{guiding principle}.

\paragraph{Classification with applied harmonic analysis.}

Certain transforms from applied harmonic analysis, like the curvelet and shearlet transform, offer an alternative possibility to identify the wavefront set. In particular, the connection between the behavior of these transforms, and the wavefront set has been analyzed in \cite{CANDES2005162, kutyniok2009resolution}, with a particular application for edge detection in \cite{yi2009shearlet}. We will recall these results in Section~\ref{sec:WaveFrontSet}. These approaches characterize the wavefront set through the rate of decay of the respective transforms. In this way, one can transform the problem of extracting the wavefront set of a function to a classification problem on the decay of another function. 
While this point of view certainly makes the wavefront set more accessible, especially since it does not depend on an unspecified localization procedure, it is still subject to the fundamental restrictions of the previous subsection, i.e., it cannot produce a successful wavefront set extractor acting on sampled functions.

\paragraph{Data driven wavefront set extraction.}
Bearing in mind the aforementioned guiding principle, a successful wavefront set extractor needs to be tailored to the function class of interest. The relevant function classes in applications are, however, difficult to characterize analytically. An alternative is, therefore, to adopt a \emph{data-driven model} where the function class of interest is given empirically through examples.

Based on the above, we propose the following algorithm coined DeNSE: First, we assemble a supervised training set consisting of images with their associated wavefront set, or a suitable surrogate. Second, we train a classifier---in our case a \emph{deep neural network}---to predict the wavefront set from the \emph{shearlet coefficients} of the training data. 
Finally, we apply the resulting classifier on unseen data.  

In this way, the algorithm combines two crucial elements: On the one hand, as mentioned previously, the interaction of the shearlet transform with the wavefront set is theoretically well understood and presents the microlocal information of a function in a more accessible way. On the other hand, the trained classifier allows a strong adaptation to the underlying function class, thereby complying with our guiding principle.

In Section~\ref{sec:Algorithm}, we present the construction of the algorithm alongside the training data that is used. In this context, we analyze the detection of edges and orientations in images as well as higher-order wavefront set detection in sinogram data.
We shall see below, in Section~\ref{sec:numRes}, that this method outperforms all conventional edge-orientation estimators as well as alternative data-driven methods including the current state-of-the-art. Moreover, we are unaware of any wavefront set extractor in the literature that goes beyond edge or ramp detection, so that the following analysis can be seen as the first advance in this direction. 

\subsection{Expected impact}
The DeNSE method has many applications, some of which are outlined below:

\begin{itemize}
\item \emph{Solving inverse problems:} Many inverse problems involve a pseudo-differential or a Fourier integral operator. Moreover, the associated microlocal canonical relations can be explicitly derived for such operators. This gives a relation between the wavefront sets of measurement data and the unknown signal one seeks to recover. In particular, one may use DeNSE to extract the wavefront set of data and then use the relevant microlocal canonical relation to compute the wavefront set of the unknown signal without solving the inverse problem. Section~\ref{sec:SinogramResults} demonstrates this approach in the context of tomography.


 \item \emph{Regularization of inverse problems:} 
A priori knowledge of the relation between the wavefront sets for data and signal (the solution) can be used to regularize an ill-posed inverse problem. An example of this, which used shearlets for the identification of the wavefront set, was demonstrated in \cite{bubba2018learning} in the context of limited angle tomography. This approach is, however, specifically tailored for inverse problems involving the Radon transform. In contrast, our DeNSE based approach is more generic and applies readily to a wide range of inverse problems, see Section~\ref{sec:RadonTransform} for further details.

\item \emph{Edge detection:} Detecting edges, ridges, or points of higher-order non-smoothness is a sub-task of wavefront set extraction. As we will observe below, DeNSE outperforms state-of-the-art competing edge-orientation detectors on a wide range of test cases. Moreover, except for ridge and ramp detection, the detection of points of higher-order non-smoothness has---to the best of our knowledge---not been pursued in the literature, but is possible with DeNSE without adaptations.
\end{itemize}

\subsection{Basic concepts and notation}

Below, we collect the notation used throughout this manuscript. This notation is fairly standard in the literature, and hence this subsection can be skipped and only consulted if a notation is unclear.

$\Real$, $\Natural$, and $\Integer$ denote the set of real numbers, natural numbers, and integers, respectively.
Next, given a point $x \in \Real^n$ and $r >0$, we use $B_{r}(x)$ to denote the ball of radius $r$ in $\Real^n$ with center at $x$. Likewise, $\mathbb{S}^{n-1}$ denotes the unit sphere in $\Real^n$.
Furthermore, the boundary of a domain $\Omega \subset \Real^d$ for $d \in \Natural$ is denoted by $\partial \Omega$.

Let $\Omega \subset \Real^d$ be a fixed domain for some $d \in \Natural$. Then, $L^2(\Omega)$ is the space of Lebesgue square-integrable functions on $\Omega$, $C^n(\Omega)$ is the space of $n$-times continuously differentiable functions defined on $\Omega$, and $H^n(\Omega)$ is the space of $n$-times weakly differentiable functions whose weak derivatives are in $L^2(\Omega)$. The support of a measurable function $f\colon \Real^d \to \Real$ is denoted by $\suppp f$. Furthermore, the Fourier transform of a function $f \in L^1(\Real^d)$ is defined as
\begin{align*}
    \hat{f}(\xi) \coloneqq \int_{\Real^d} f(x) e^{-2\pi i \langle x, \xi \rangle } dx \quad \text{ for } \xi \in \Real^d.
\end{align*}
The Fourier transform operator $f \mapsto \hat{f}$ can be extended to an isometry on $L^2(\Real^n)$ by Plancherel's identity. The Paley-Wiener spaces are formed by functions the Fourier transform of which is compactly supported. More precisely, for $\Lambda \in \Real_+$, we define
\[
PW_{\Lambda} \coloneqq \bigl\{ f \in L^2(\Real^d) \colon \suppp \hat{f} \in [-\Lambda, \Lambda]^d \bigr\}.
\]
Finally, we use the Landau symbol $\mathcal{O}$ to describe asymptotic behavior, i.e., for functions $f, g \colon \Real \to \Real$, we write $f(x) = \mathcal{O}(g(x))$ as $x \to a$ whenever there exists a constant $c>0$ such that $|f(x)| \leq c |g(x)|$ for all $x$ in a neighborhood of $a$. Similarly, we write $f(x) = \mathcal{O}(g(x))$ as $x \to \infty$ whenever there exists a constant $c>0$ such that $|f(x)| \leq c |g(x)|$ for $|x|$ sufficiently large.

\section{Directional multiscale systems and the wavefront set} \label{sec:WaveFrontSet}

We start by formally introducing the notion of a wavefront set followed by the definition of the directional multiscale system of shearlets.
We then show how shearlets can indeed be used to resolve the wavefront set of a distribution. Similar results also hold for other multiscale systems, like the curvelet transform \cite{CANDES2005162} and the more general continuous parabolic molecules \cite{grohs2015continuous}.

\subsection{The wavefront set}
The wavefront set can be defined for any distribution on a manifold, but since we only deal with $L^2(\Real^2)$ functions, we restrict the definition to this setting.

\begin{definition}[{\cite[Section 8.1]{AnLinPDOHoermander}}]\label{def:WaveFrontSet}
Let $f \in L^2(\Real^2)$ and $k \in \Natural$. A point $(x, \lambda)\in \Real^2 \times \mathbb{S}^1$ is a \emph{$k$-regular directed point of $f$} if there exist open neighborhoods $U_x$ and $V_\lambda$ of $x$ and $\lambda$, respectively and a smooth function $\phi\in C^\infty(\Real^2)$ with $\suppp \phi \subset U_x$ and $\phi(x) = 1$ such that
\begin{equation*}
    \bigl| \widehat{\phi f} (\xi) \bigr|
    \leq C_k \bigl( 1+|\xi| \bigr)^{-k}
    \quad \text{for all $\xi \in \Real^2 \setminus \{0\}$ such that $\xi /|\xi| \in V_\lambda$}
\end{equation*}
holds for some $C_k >0$. The \emph{$k$-wavefront set} $\WF_k(f)$ is the complement of the set of all $k$-regular directed points and the \emph{wavefront set} $\WF(f)$ is defined as
\[
\WF(f) \coloneqq \bigcup_{k \in \Natural} \WF_k(f).
\]
\end{definition}
The definition of the wavefront set is based on the well-known characterization of smoothness of a function in terms of the decay of its Fourier transform. More precisely, a function is smooth at a point if its Fourier transform decays faster than any polynomial in any direction.
As an example, the \emph{singular support} of $f$, i.e., the smallest closed set $U$ such that $f_{|U^c} \in C^\infty(U^c)$ can be characterized in terms of the wavefront set as
\[
\bigl\{ x \in \Real^2 : (x, \lambda) \in \WF(f)\text{ for some $\lambda \in \mathbb{S}^1$} \bigr\}.
\]
The wavefront set is a refined notion of the singular support, since it not only indicates at which points a function is not smooth, but also contains the associated directions causing the non-smoothness. A common example considers $f$ with a jump singularity across the smooth boundary of a fixed domain $D \subset \Real^2$. Then, one can show \cite[Chapter VI, Exercise 1.1]{taylor1981pseudodifferential} that
\begin{equation}\label{eq:wavefrontSet}
\WF(f) = \bigl\{(x, \lambda) \in \Real^2\times \mathbb{S}^1 \colon x \in \partial D \text{ and } \lambda = n_x \text{ where $n_x$ is a normal on $\partial D$ at $x$}
\bigr\}.
\end{equation}

The wavefront set is a very powerful tool in mathematical analysis, but it is difficult to compute in practice. This is mainly due to the asymptotic criteria involved in its definition, which means  computing the wavefront set requires computing the ``full'' Fourier transform at every point.
Continuous transforms associated to certain directional multiscale systems offer a convenient remedy. As an example, the shearlet transform automatically performs the necessary time-frequency-orientation localization described above, thereby ``resolving'' the wavefront set in a sense described in the following subsection.

\subsection{Shearlets}\label{sec:Shearlets}
The shearlet transform, which was introduced in \cite{gitta2005shearlets}, is based on applying translation, anisotropic dilation, and shearing to generator functions. To dilate and shear a function, we define the following three matrices:
\begin{equation*}
    A_a\coloneqq \begin{pmatrix}
        a & 0 \\
        0 & \sqrt{a}
    \end{pmatrix},
    \quad
    \widetilde{A}_a\coloneqq \begin{pmatrix}
        \sqrt{a} & 0 \\
        0 & a
    \end{pmatrix},
    \quad\text{and}\quad
    S_s \coloneqq \begin{pmatrix}
        1 & s \\
        0 & 1
    \end{pmatrix}
    \quad\text{for $a>0$ and $s\in \Real$.}
\end{equation*}
Next, given $(a,s,t) \in \Real_+ \times \Real \times \Real^2$, $\psi \in L^2(\Real^2)$, and $x\in \Real^2$, we define
\begin{equation}\label{eq:shearlets}
\psi_{a,s,t, 1}(x) \coloneqq a^{-\frac{3}{4}} \psi\left(A_a^{-1}S_s^{-1} (x-t)\right)
\quad\text{and}\quad
\psi_{a,s,t, -1}(x) \coloneqq a^{-\frac{3}{4}} \widetilde{\psi}\left(\widetilde{A}_a^{-1}S_s^{-T} (x-t)\right),
\end{equation}
where $\widetilde{\psi}(x_1,x_2) \coloneqq \psi(x_2,x_1)$ for all $x = (x_1,x_2) \in \Real^2$.
Following \cite{grohs2011continuous}, we define the continuous shearlet transform as follows:
\begin{definition}[Continuous shearlet transform]
Let $\psi \in L^2(\Real^2)$. Then the family of functions $\psi_{a,s,t, \iota} \colon \Real^2 \to \Real$ parametrized 
by $(a,s,t, \iota) \in \Real^+ \times \Real \times \Real^2 \times \{-1,1\}$ that are defined in \eqref{eq:shearlets} is called a \emph{shearlet system}.
The corresponding \emph{(continuous) shearlet transform} is defined by
\[
    \mathcal{SH}_\psi: L^2(\Real^2) \to L^\infty\bigl(\Real^+ \times \Real \times \Real^2 \times \{-1,1\}\bigr)
    \quad\text{where}\quad
    \mathcal{SH}_\psi(f)(a,s,t,\iota) \coloneqq \langle f, \psi_{a,s,t,\iota}\rangle.
\]
\end{definition}
As we shall see next, if the generator function $\psi$ has directional vanishing moments, then the asymptotic behavior as $a \to 0$ of the continuous shearlet transform of an $L^2$-function $f$ characterizes its wavefront set.
The precise statement in \cite{grohs2011continuous} reads as follows.

%
\begin{theorem}\label{thm:charOfWavefront}
Let $f \in L^2(\Real^2)$ and assume $(x_0, \lambda_0) \in \Real^2 \times \mathbb{S}^1$ is a $k$-regular directed point of $f$ for some $k \in \Natural$.
Next, consider a continuous shearlet system with generator function $\psi \in H^l(\Real^2)$, $l \in \Natural$, with Fourier transform $\widehat{\psi} \in L^1(\Real^2)$ where $\psi$ has $m \in \Natural$ vanishing moments in $x_1$-direction, i.e.,
\[
\int_{\Real^2} \frac{\bigl| \widehat{\psi}(\xi_1,\xi_2) \bigr|^2}{|\xi_1|^{2m}} d \xi < \infty.
\]
Finally, assume $\psi$ displays the following asymptotic behavior for $p \in \Natural$:
\[
\bigl| \psi(x) \bigr| = \mathcal{O}\bigl((1+|x|)^{-p}\bigr)
\quad\text{for $|x| \to \infty$.}
\]
Then, there exist a neighborhood $U_{0} \subset \Real^2$ of $x_0$ and a neighborhood $S_{0} \subset \mathbb{S}^1$ of $\lambda_0$ such that
\[
\bigl| \mathcal{SH}_\psi(f)(a,s,x,\iota) \bigr|
  = \mathcal{O}\Bigl(
      a^{\frac{p}{2} - \frac{3}{4}} + a^{\frac{m}{4}} + a^{\frac{3k}{4} - \frac{3}{4}} + a^{\frac{3l}{4}}
    \Bigr)
\quad\text{as $a \to 0$}
\]
for all $x \in U_{0}$ and all $s\in \Real$ and $\iota \in \{-1,1\}$ such that $\lambda(s, \iota) \in  S_{0}$, where \begin{equation}\label{eq:LambdaEq}
 \lambda(s, \iota) \coloneqq
     \begin{cases}
       \left( \dfrac{1}{\sqrt{s^2+1}}, \dfrac{s}{\sqrt{s^2+1}} \right) & \text{if $\iota = 1$,}
       \\[1em]
       \left( \dfrac{s}{\sqrt{s^2+1}}, \dfrac{1}{\sqrt{s^2+1}} \right) & \text{if $\iota = -1$.}
     \end{cases}
\end{equation}
\end{theorem}
\begin{remark}\label{rem:InverseTheorem}
Under suitable assumptions on the shearlet generator $\psi$, the converse of Theorem~\ref{thm:charOfWavefront} holds as well.
More precisely, following \cite{grohs2011continuous}, assume that $\psi$ is sufficiently regular for any $k \in \Natural$.
Next, let $(x_0,\lambda_0) \in \Real^2 \times \mathbb{S}^1$ and assume there exist a neighborhood $U_{0} \subset \Real^2$ of $x_0$ and a neighborhood $S_{0} \subset \mathbb{S}^1$ of $\lambda_0$ such that
\[
\bigl| \mathcal{SH}_\psi(f)(a,s,x,\iota) \bigr| = \mathcal{O}(a^{n})
\quad\text{as $a \to 0$},
\]
holds for sufficiently large $n \in \Natural$ uniformly for $x \in U_0$ and all $s\in \Real$, $\iota \in \{-1,1\}$ such that $\lambda(s, \iota) \in S_{0}$.
Then, $(x_0,\lambda_0)\not \in \WF_k(f)$.
\end{remark}
Theorem~\ref{thm:charOfWavefront} and Remark~\ref{rem:InverseTheorem} demonstrate that the wavefront set is completely determined by the decay properties of the shearlet transform. This implies that in the continuous setting, one can compute the wavefront set of a function by first computing its continuous shearlet transform, then analyzing the pairs of point and direction where this shearlet transform exhibits rapid decay as $a \to 0$.

Theorem~\ref{thm:charOfWavefront} and Remark~\ref{rem:InverseTheorem} were first reported in \cite{kutyniok2009resolution} in a setting restricted to a specific shearlet generator (called the  ``classical shearlet'').
Moreover, results similar to Theorem~\ref{thm:charOfWavefront} and Remark~\ref{rem:InverseTheorem} were obtained in  \cite{CANDES2005162} for the curvelet transform and in \cite{fell2016resolution} for transforms stemming from general group representations.
As shown in \cite{grohs2015continuous}, all transforms that belong to the category of continuous parabolic molecules admit a similar characterization of the wavefront set.
Finally, one can also use the shearlet transform to classify certain geometric properties of the singularities of a function that goes beyond differentiating between rapid and non-rapid decay of the shearlet transform, see e.g. \cite{yi2009shearlet,guo2009edge,kutyniok2017classification}.

\section{Wavefront set {of sampled functions}}\label{sec:noWaveFrontSetExtractor}
{Here we analyze whether} it is possible to construct an operator that maps a {finitely sampled function} $f$ to an estimate of {its} wavefront set.
This {typically arises in} practical applications, e.g., images are only {given} as pixels representing point samples of a real-valued function.

It is natural to use Shannon's sampling theorem to make the connection between a sampled function and its wavefront set more precise. The theorem is stated in Subsection~\ref{sec:Shannon} and -- based on it --  {Subsection~\ref{sec:NotOfWave} introduces} the notion of an approximate wavefront set extractor. Finally, in Subsection~\ref{sec:NoWaveExtr} we show that any approximate wavefront set extractor on a function class $\mathcal{C}\subset L^2(\Real^2)$ 
that predicts the wavefront set of a function $f \in \mathcal{C}$ from a finite number of sample values will fail on a dense subset of $\mathcal{C}$ if 
$\mathcal{C}$ contains at least all $C^k(\Real^2) \cap L^2(\Real^2)$ functions for an arbitrary $k \in \Natural$. This result holds even if the sampling density is allowed to depend on the function $f$.

\subsection{Sampling theorem and Paley-Wiener spaces}\label{sec:Shannon}
The sampling theorem states that every band-limited function $f$ can be written as a sum of shifted cardinal sine functions weighted by point samples of $f$. In other words, a band-limited function is fully determined by its values on a discrete grid.
To give the precise statement, we introduce the Paley-Wiener spaces. Given $\Lambda>0$, the Paley-Wiener space $\PW_{\Lambda} \subset L^2(\Real^d)$ is defined as
\[
  \PW_{\Lambda} \coloneqq
  \left\{ f \in L^2\left(\Real^d\right) \colon \supp\big(\widehat{f} \, \big) \subset [-\Lambda,\Lambda]^d \right\}.
\]
We define the $d$-dimensional sinc-function as
\[
\sinc_{d}(x) \coloneqq \prod_{i=1}^d  \frac{\sin(\pi x_i)}{\pi x_i},
\quad\text{where } x = (x_1, \dots, x_d) \in \Real^d.
\]
Bearing in mind the above notation, we now state the sampling theorem, see, e.g., \cite{marks2012introduction}.
\begin{theorem}[Sampling theorem] \label{thm:SamplingTheorem}
Let $d \in \Natural$, $f \in L^2(\Real^d)\cap C(\Real^d)$ and $\Lambda>0$. Then
\[
f \in \PW_{\Lambda}
\iff
f(x) = \sum_{n \in \Integer^d} f\left(\frac{n}{\Lambda}\right)
  \sinc_{d}(\Lambda \cdot x- n )\quad \text{ for all $x \in \Real^d$.}
\]
\end{theorem}
In particular, for every summable sequence $(y_k)_{k\in \Integer^2}$ we can define
\[
    f(x) \coloneqq \sum_{k \in \Integer^2} y_k \, \sinc_{2}(\Lambda \cdot x -k),\quad \text{ for }x \in \Real^d,
\]
and, by Theorem~\ref{thm:SamplingTheorem}, $f$ is band-limited.
Furthermore, since $\sinc_2(m-n)$ vanishes for every $m, n \in \Integer^2$ such that $m \neq n$, we observe that $f(m/\Lambda) = y_m$ for all $m \in \Integer^2$.
In other words, every sequence on a grid defines an associated interpolating band-limited function and conversely, every band-limited function is uniquely determined by its values on a discrete grid.

As a consequence, the problem of extracting the wavefront set of a function $f\in L^2(\Real^2) \cap C(\Real^2)$ from its discrete sampled values $(f(m/\Lambda))_{m\in \Integer^2}$ can be re-stated as extracting the wavefront set of $f$ from its projection onto a Paley-Wiener space, i.e., $P_{\PW_{\Lambda}}(f)\eqqcolon P_\Lambda(f)$.

For functions $f \in L^2(\Real^2)$, which are only defined almost everywhere, we can even use the sampling theorem as a definition of a point evaluation. For $f \in L^2(\Real^2)$, we have that $P_{\Lambda}(f)\in C(\Real^2)$. Hence, we set $f(m/\Lambda) \coloneqq P_{\Lambda}(f)(m/\Lambda)$.

\subsection{Wavefront set extractors}\label{sec:NotOfWave}
As already stated, the problem of extracting the wavefront set from samples on a grid is equivalent to extracting the wavefront set given the projection onto a Paley-Wiener space. There are multiple conceivable notions of a wavefront set extractor. First, for $\Lambda >0$, we could ask for a map
\begin{align}\label{eq:OneShotExtractor}
\DWF_\Lambda: \PW_\Lambda \to P\left(\Real^2 \times \mathbb{S}^1 \right) \quad\text{such that $\DWF(P_{\Lambda} f ) = \WF(f)$ for all $f \in L^2(\Real^2)$.}
\end{align}
Here $P(\Real^2 \times \mathbb{S}^1)$ denotes the power set of $\Real^2 \times \mathbb{S}^1$.
Essentially, this map requests extraction of the wavefront set of a function $f$ from knowledge of the samples of $f$ on a fixed grid. It is clear that such a map, $\DWF_\Lambda$, cannot exist, since it will fail for functions $f$ that have fine structures which cannot be detected by coarse sampling.
For example, a function that vanishes on every grid point of $\Integer^2/\Lambda$ while having a non-trivial wavefront set would be classified the same as the zero function.

A more reasonable model for a wavefront set extractor should give an approximate prediction of the wavefront set 
that eventually improves as the sampling density increases.
To weaken this statement even further, we might only ask for approximate extraction of the wavefront set at one point and only for functions from a fixed function class $\mathcal{C}\subset L^2(\Real^2)$.
For a fixed set $W \subset \Real^2 \times \mathbb{S}^1$ and a point $x \in \Real^2$, we define
\[
    W_x \coloneqq \bigl\{ \lambda \in \mathbb{S}^1 : (x, \lambda) \in W \bigr\}.
\]
We can now model the approximation described above by considering a sequence of wavefront set extractors given by
\begin{equation}\label{eq:SeqOfWFExtr}
  \DWF_j \colon \PW_j \to \mathcal{P}\left({\Real^2 \times \mathbb{S}^1}\right)
  \quad\text{for $j\in \Natural$}
\end{equation}
such that, for fixed $x \in \Real^2$, and all $f \in \mathcal{C}$,
\begin{align}\label{eq:clairvoyance}
d_H\left( \overline{\DWF_j( P_j(f) )_x}, \overline{\WF(f)_x} \right) \to 0.
\end{align}
Here $d_{H}$ denotes the Hausdorff distance with the convention $d_H(X, \emptyset) = d_H(\emptyset, X) := 1$ for any non-empty compact subset $X \subset \mathbb{S}^1$ and $d_H(\emptyset,\emptyset) := 0$.
Recall that with this definition $d_{H}$ is a metric on compact subsets of $\mathbb{S}^1$ (including the empty set).

A sequence as in \eqref{eq:SeqOfWFExtr} satisfying \eqref{eq:clairvoyance} yields an approximate extraction of the wavefront set of $f$ at $x$ from point samples of $f$ where the sampling density may depend on $f$. This observation motivates the following definition.

\begin{definition}
Let $\mathcal{C} \subset L^2(\Real^2)$. A sequence $(\DWF_j)_{j\in \Natural}$ of mappings as in \eqref{eq:SeqOfWFExtr} is called an \emph{approximate wavefront set extractor}. We say that an approximate wavefront set extractor is
\begin{itemize}
\item \emph{clairvoyant} at $x \in \Real^2$ if the sequence satisfies \eqref{eq:clairvoyance} at $x$ for all $f \in \mathcal{C}$, and
\item \emph{ignorant} to $f \in \mathcal{C}$ at $x \in \Real^2$ if $d_H(\overline{\DWF_j( P_j(f) )_x}, \overline{\WF(f)_x}) \not \to 0$ as $j \to \infty$.
\end{itemize}
\end{definition}

\subsection{Non-existence of clairvoyant approximate wavefront set extractors}\label{sec:NoWaveExtr}

We will observe below that, for every $x \in \Real^2$, there does not exist a clairvoyant approximate wavefront set extractor if $\mathcal{C}$ contains at least all $k$-times differentiable functions, for a $k \in \Natural$. Even more severely, in this situation, every approximate wavefront set extractor is ignorant to a dense subset of $\mathcal{C}$ at $x$. 

\begin{theorem}\label{thm:NoClarevoyance}
Let $k \in \Natural$ and $\mathcal{C} \subset L^2(\Real^2)$ be such that $\mathcal{C} \supset C^k(\Real^2) \cap L^2(\Real^2)$. For every $x \in \Real^2$ we have that, for every approximate wavefront extractor $(\DWF_j)_{j\in \Natural}$, there exists a dense subset $\mathcal{M} \subset \mathcal{C}$ such that $(\DWF_j)_{j\in \Natural}$ is ignorant to all $f \in \mathcal{M}$ at $x$. In particular, no approximate wavefront set extractor is clairvoyant at $x$.
\end{theorem}
\begin{proof}
The proof proceeds in two steps. First, for a given approximate wavefront set extractor, $(\DWF_j)_{j\in \Natural}$, and a point $x\in \Real^2$, we construct a function $q \in\mathcal{C}$ such that $(\DWF_j)_{j\in \Natural}$ is ignorant to $q$ at $x$.
Second, we show that the set of such functions is dense in $\mathcal{C}$.

\smallskip

\textbf{Step 1:} Notice that Definition~\ref{def:WaveFrontSet} implies
\begin{equation}\label{eq:sumsOfWFSets}
    \WF(f_1+f_2) = \WF(f_1)
    \quad\text{for every $f_1 \in L^2(\Real^2)$ and $f_2 \in C^\infty(\Real^2) \cap L^2(\Real^2)$.}
\end{equation}
For $x \in \Real^2$, we now choose a function $g \in C^k\cap L^2(\Real^2)$ such that for a non-empty set $V$ we have that $\WF(g)\supset \{x\} \times V$. Since $k < \infty$, such a function always exists. We can also assume hat $\|(1+|\cdot|^2)^{k/2}\widehat{g}\|_{L^1(\Real^2)} < \infty$.

Then, by \eqref{eq:sumsOfWFSets}, we can conclude that $\WF(g - P_{j}g) \supset \{x\} \times V$ holds for every $j \in \Natural$. Moreover, by construction, we have $d_H(\overline{\WF(g)_x}, \emptyset) = 1$. To define the desired function $q$, we first set
\begin{align}\label{eq:constructionOfqn}
q_0 &\coloneqq P_{1} g, \nonumber
\\[0.75em]
q_n &\coloneqq
  \begin{cases}
    q_{n-1} +  (P_{n} g - P_{n-1} g) &\text{if $\DWF_{j-1}(P_{n-1}q_{n-1})_x = \emptyset$,}
    \\[0.5em]
    q_{n-1} & \text{otherwise,}
\end{cases}
\end{align}
for all $n \ge 1$. We have that by the Riemann-Lebesgue Lemma, for every $N \in \Natural$,
\begin{equation} \label{eq:summabilityProperty}
\sum_{n \leq N} \|q_n - q_{n-1}\|_{C^k} \lesssim \sum_{n \leq N} \left\|(1+|\cdot|^2)^{\frac{k}{2}}\left({\widehat{q}_n - \widehat{q}_{n-1}}\right)\right\|_{L^1}  \leq  \left\|(1+|\cdot|^2)^{\frac{k}{2}} \hat{g}\right\|_{L^1} < \infty.
\end{equation}
Moreover, by definition,
\[
    (P_{n} g - P_{n-1} g) \perp (P_{m} g - P_{m-1} g)
    \quad\text{for all $n \neq m$.}
\]
Hence, by the Pythagorean theorem, for every $N \in \Natural$,
\begin{equation} \label{eq:summabilityProperty2}
  \sum_{n \leq N} \|P_{n} g - P_{n-1} g\|_2^2 \leq \|g\|_2^2 < \infty.
\end{equation}
It now follows from \eqref{eq:summabilityProperty}, \eqref{eq:summabilityProperty2} and \eqref{eq:constructionOfqn} that $q_n$ is a Cauchy sequence in $C^k(\Real^2)$ and $L^2(\Real^2)$. Therefore $q_n$ converges to a limit $q \in L^2(\Real^2) \cap C^k(\Real^2) \subset \mathcal{C}$.
Furthermore, one of the following statements holds:
\begin{enumerate}
\item[(1)] $\overline{\DWF_j(P_{j}q_{j})_x}$ does not converge for $j \to \infty$;
\item[(2)] $\overline{\DWF_j(P_{j}q_{j})_x}$ converges to a limit $W$ such that $d_H({W}, \overline{\WF(q)_x}) \geq 1/2$;
\item[(3)] $\overline{\DWF_j(P_{j}q_{j})_x}$ converges to a limit $W$ such that $d_H({W}, \overline{\WF(q)_x}) < 1/2$.
\end{enumerate}
In Cases~(1) and (2), we directly obtain that $\DWF_{j}$ is ignorant to $q$ at $x$.
In Cases~(3), we obtain that there exists some $j_0$ such that
\begin{align}\label{eq:Case3Conclusion}
    d_H\bigl(\overline{\DWF_j(P_{j}q_{j})_x},\overline{\WF(q)_x} \bigr) < 1
\quad\text{for all $j \geq j_0$.}
\end{align}
We now consider the cases $\WF(q)_x  = \emptyset$ and $\WF(q)_x \neq \emptyset$ separately.
If $\WF(q)_x  = \emptyset$, then \eqref{eq:Case3Conclusion} implies that $\DWF_j(P_{j}q_{j})_x = \emptyset$ for all $j \geq j_0$ since no subset of $P(\mathbb{S}^1) \setminus \{\emptyset\}$ has a distance less than $1$ to the empty set. Therefore,
\begin{align}\label{eq:Case3SecondConclusion}
q-P_{j_0}q = \sum_{j > j_0} (P_{j} g - P_{j-1} g) = g - P_{j_0} g. 
\end{align}
We obtain from \eqref{eq:Case3SecondConclusion} that $\emptyset = \WF(q)_x = \WF(g)_x  \neq \emptyset$ which is a contradiction.

If $\WF(q)_x \neq \emptyset$, then $d_H(\overline{\WF(q)_x}, \emptyset) = 1$.
By the triangle inequality, this yields that there exists some $j_0$ such that $\DWF_j(P_{j}q_{j})_x \neq \emptyset$ for all $j \geq j_0$.
Therefore, $q = q_{j_0} \in \PW_{j_0}$ by definition, which implies that $W(q)_x = \emptyset$.
Hence, Case (3) does not occur, i.e., $(\DWF_j)_{j\in \Natural}$ is ignorant to $q$ at $x$.

\textbf{Step 2:} For an arbitrary $f \in \mathcal{C}$, there exists $j_1 \in \Natural$ such that
\[
\left\|f - P_{j_1} f \right\|_2 \leq \frac{\epsilon}{2}
\quad\text{and}\quad
\left\|g - P_{j_1} g \right\|_2\leq \frac{\epsilon}{2}.
\]
Define $q_{j_1} = P_{j_1} f$ and, for every $n \geq j_1$, we define $q_n$ as in \eqref{eq:constructionOfqn}. It is clear that $q_n$ converges to a limit $q_f\in C^k(\Real^2) \cap L^2(\Real^2)$. Also, it is straightforward to show that $\|q_f-f\|_{2}\leq \epsilon$.
Now using the same arguments as in Step~1, it follows that $(\DWF_j)_{j\in \Natural}$ is ignorant to $q_f$.
\end{proof}

\begin{remark}
\begin{itemize}
\item[(1)] The Theorem~\ref{thm:NoClarevoyance} and its proof also hold when ``wavefront set'' is replaced by ``$(k+1)$-wavefront set'' or ``singular support''. 
\item[(2)] The arguments in the proof of Theorem \ref{thm:NoClarevoyance} are independent from the domain $\Real^2$.
Indeed, the same result holds for functions defined in an open domain $\Omega \subset \Real^2$ and $x \in \Omega$.
Here we define the wavefront set of $f\in L^2(\Omega)$ as
\[
    \Bigl\{ (x,\lambda) \in \Omega \times \mathbb{S}^1 : (x,\lambda) \in \WF\big(\tilde{f}\, \big)
    \text{ where $\tilde{f} = f$ on $\Omega$ and $\tilde{f} = 0$ elsewhere}
    \Bigr\}.
\]
\item[(3)] Theorem~\ref{thm:NoClarevoyance} demonstrates that there does not exist any clairvoyant approximate wavefront set extractor for sufficiently large function classes.
Even more severely, for every $k \in \Natural$, every approximate wavefront set extractor fails on a dense subset of $L^2(\Real^2) \cap C^k(\Real^2)$.

On the other hand, if the function class is so small that every function is uniquely determined by its values on the grid, then one can construct a wavefront set extractor by storing the wavefront set for each function in a data base.

This shows that for most classical function classes it is impossible to analytically derive a wavefront set extractor. For function classes in applications, such as sets of images, wavefront set extractors could exist, but are potentially highly sensitive to the choice of $\mathcal{C}$.
\end{itemize}

\end{remark}

\section{Wavefront set and edge detectors}
\label{sec:WFset-edge}
The characterization of a wavefront set given in \eqref{eq:wavefrontSet} implies that detecting the wavefront set of a piece-wise smooth function with singularities along a smooth curve is equivalent to detecting edges and their normal direction.
Edge detection, which is one of the most well-studied problems in image processing, is therefore a sub-problem to wavefront set extraction. However, we are unaware of any wavefront set extractor in the literature that goes beyond edge or ramp detection.

We next recall some edge-orientation detectors from the literature.
We start by approaches based on filters, which is followed by a review of methods based on directional systems from applied harmonic analysis. We then compare these methods to the guiding principle. Finally, we will comment on recent data-driven algorithms for edge detection.

\subsection{Filter-based edge-orientation detectors}
The traditional way of detecting edges in digital images is to convolve the image with suitable convolution kernels to enhance edge-like features. These features can then be extracted using simple rules.
For example, convolution with local difference filters leads to the Roberts \cite{roberts1963machine}, Sobel \cite{duda1973pattern, sobel2014edge}, and Prewitt \cite{prewitt1970object} operators.
In \cite{marr1980theory} the authors convolve an image with the Laplacian of a Gaussian function. Here, the zeros of the resulting image are taken as estimate for the positions of edges in the original.
Potentially the most famous edge-orientation detection algorithm of this category is the Canny edge detector \cite{canny1986computational}, which is based on convolution with a Gaussian kernel with standard deviation $\sigma>0$. The obvious drawback of this algorithm is that the choice of the standard deviation $\sigma$ of the Gaussian window strongly influences the performance of the algorithm to the extent that a high $\sigma$ improves the robustness against noise, but might also remove high-frequency components.

There have been many efforts in addressing the shortcomings of the Canny edge detector.
In \cite{perona1990detecting} one extracts the positions and orientations of not only jump singularities but also of composite edges like ramp- and hat-like singularities.
Most recently, \cite{arbelaez2011contour} develops algorithms that apply higher-level heuristics to a set of patches retrieved from applying oriented gradient operators.

\subsection{Edge-orientation detectors based on directional systems}
An alternative to the filter-based approaches is to base the detector on multiscale directional transforms.
These transforms include shearlets, which were introduced in Subsection \ref{sec:Shearlets}, as well as curvelets, ridgelets, wavelets, or bandlets, see  \cite{jacques2011panorama} for a survey.

The idea is to first transform the given image using a directional system.
Next, one applies certain heuristics to the result by using the theoretical information on the behavior of the underlying transform at directed edge points.
In principle, the transforms given by directional systems mentioned above can be written as a series of convolutions with respect to directed filters.
Hence, this approach can, in a sense, be understood as a special case of the edge detectors in the previous section.

An example of such an approach, which served as the main inspiration of the forthcoming algorithm in Section~\ref{sec:Algorithm}, is the shearlet-based algorithm of \cite{yi2009shearlet}, which we will now explain in more detail. This algorithm seeks to detect and classify edges of {piece-wise} smooth functions $f$ with {piece-wise} smooth jump curves.
The algorithm is based on computing a sampled version of the shearlet transform of $f$, i.e., the algorithm computes
\begin{equation}\label{eq:ShearletCoeffs}
    \Bigl\{\mathcal{SH}(f)(2^{-j}, 2^{-\frac{j}{2}} k, x, \iota): j\in \Natural, k\in \Integer, x \in c\Integer^2, j_0 \leq j \leq J, |k| \leq 2^{\frac{j}{2}}, \iota \in \{-1,1\} \Bigr\}
\end{equation}
for given $j_0, J\in \Natural$ and $c>0$.
This algorithmic step can be performed using standard software libraries such as the ShearLab toolbox \cite{Kutyniok:2016:SFD:2888419.2740960}, see also Subsection~\ref{sec:DigShearTrafo}.
This computation is followed by a series of heuristics that are applied to the coefficients \eqref{eq:ShearletCoeffs} and which lead to a classification of points as point singularities, directed edges, corners, or smooth points.

Additionally, we mention an algorithm in \cite{rafael2015coshrem} that is based on the complex shearlet transform.
This algorithm computes two shearlet transforms, one with a symmetric and one with an anti-symmetric generator.
The relationship between both is then used to determine if a point is an edge, a ridge, or a smooth point.

\subsection{Data-driven edge-orientation detectors}
As we have already advocated previously, the only way to comply with the guiding principle in applications is to use a data-driven algorithm. For edge detection, this idea has already been followed in a series of papers. For example, \cite{bel2006edge} uses supervised learning of an edge classifier based on a technique called probabilistic boosting tree.
Another approach is DeepEdge that uses a convolutional neural network taking as an input candidate edges produced by a Canny edge detector as well as patches of the original image \cite{bertasius2015deepedge}.
To this input one then applies parts of the KNet \cite{krizhevsky2012imagenet} for feature extraction and a network with two fully-connected layers for classification.
Finally, the algorithms SEAL (simultaneous edge alignment and learning) \cite{yu2018simultaneous} and the CASENet (category-aware semantic edge detection network) \cite{yu2017casenet} perform high-level edge analysis with highly complex deep neural networks. The underlying CASENet is a 101-layered network. These methods are state-of-the-art for segmentation and edge detection.

\section{Higher-order wavefront set extraction and inverse problems}
\label{sec:RadonTransform}

The discussion of the previous section was limited to edge-orientation detection. In the context of certain inverse problems, also the detection of higher-order wavefront sets is relevant. In particular, this is the case for inverse problems, whose forward operator is either pseudo differential or Fourier integral. We will focus on the so-called X-ray transform which is the basis of the computed tomography (CT) imaging technique.

The goal of computed tomography (CT) is to get an image of the internal structure of an object by striking X-rays through the object coming from a known source. The rays travel on a  line $L$ from the source through the object to a detector, being attenuated by the material on the specific line. One can mathematically model this measuring procedure by defining the Radon transform:

\begin{definition}[Radon transform, \cite{quinto2012introxray}]
Let $f\in L^1(\mathbb{R}^2)$, $\varphi\in [0,\pi)$ and $s\in \mathbb{R}^2$. Moreover, the unit vector in direction $\varphi$ denoted as $\theta$ and its orthogonal vector $\theta^{\perp}$ are given by:
$$
\theta = \theta(\varphi) := (\cos\varphi, \sin\varphi) \quad \theta^{\perp} = \theta^{\perp}(\varphi) := (-\sin\varphi, \cos\varphi).
$$
The \emph{Radon transform} of $f$ is given by the line integral 
\begin{equation}
\label{eq:RayTransform}
\mathcal{R}f(s, \varphi) := \int_{x\in L(\varphi, s)} f(x)dx = \int_{-\infty}^{\infty} f(s\theta(\varphi)+t\theta^{\perp}(\varphi))dt,
\end{equation}
where the line $L(\varphi, s)$ is defined as
$$
L(s, \varphi):= \{ x\in\mathbb{R}^2: x\cdot\theta(\varphi) = s\}.
$$
\end{definition}

We will refer to $g\coloneqq \mathcal{R}(f)$ as the \emph{sinogram} of $f$. The task of CT reconstruction is to recover the function $f$ from noisy samples of its sinogram. This task is a classic example of an ill-posed inverse problem \cite{devison1983limited}. 
{The standard approach to overcome the ill-posedness is to incorporate a priori information, an approach that is known as \emph{regularization}}. 
There are plenty of regularization methods for computed tomography ranging from {established} model-based methods \cite{CTmode1, CTmodel2, CTmodel3}, to {more recent end-to-end data-driven approaches, for the latter see the survey in \cite{Arridge:2019aa}.}
Additionally, the sparsity of the solution under the shearlet representation \cite{bubba2018learning} or the adjoint operator \cite{Hybrid1} can be used for regularization. 
Below, we outline a general approach to implement a priori knowledge of the wavefront of the sinograms into  a regularization procedure.

Assuming for the moment that we could successfully build a digital wavefront set extractor. In the context of the tomography problem, this algorithm could be applied to the sampled sinogram $g$. To compute from this knowledge the wavefront set of $f$, we can invoke a certain relationship between the wavefront set of $f$ and its sinogram $g$. 
This relation between the wavefront set of a distribution and the wavefront set of its transformation under a Fourier integral operator (FIO) is known as \textit{micro-canonical relation}. The Radon transform $\mathcal{R}$ and its adjoint $\mathcal{R}^*$ are examples of FIOs (see e.g. \cite{quinto2012introxray}), and its micro-canonical relation is described in the next theorem. 

\begin{theorem}[Radon transform micro-canonical relation, \cite{quinto2012introxray}]
\label{thm:Radoncanon}
Let $f$ be a distribution with compact support. Let $x_0\in L(s_0,\varphi_0)$, $\theta_0 = \theta(\varphi_0)$, $\lambda_0 = (1,-x_0\cdot\theta_0^{\perp})/(1+(x_0\cdot\theta_0^{\perp})^2)^{1/2}$ and $k\in \mathbb{N}_0$. The \emph{micro-canonical relation} is the wavefront set correspondence
\begin{equation}
    \label{eq:Canonical}
    (x_0 ; \theta_0)\in \WF_k(f) \Longleftrightarrow ((s_0, \varphi_0); \lambda_0) \in \WF_{k+1/2}(\mathcal{R}f)
\end{equation}
in particular, for $\WF(f) = \bigcup_{k \in \mathbb{N}_0} \WF_k(f)$ it holds
$$
(x_0 ; \theta_0)\in \WF(f) \Longleftrightarrow ((s_0, \varphi_0); \lambda_0) \in \WF(\mathcal{R}f).
$$
\end{theorem}
The relation \eqref{eq:Canonical} naturally induces an invertible map $\mathbf{can}: \Real^2 \times \mathcal{P}(\mathbb{S}^1) \to \Real^2 \times \mathcal{P}(\mathbb{S}^1)$ such that 
\begin{align}\label{eq:CanonicalMap}
\WF(\mathcal{R}f) = {\mathbf{can}}(\WF(f))
\end{align}
holds for all distributions $f$ with compact support.

After extracting the wavefront set of the sinogram one can use it in a regularization functional to regularize the inverse problem, i.e. minimize jointly the residual $\|\mathcal{R}f-g\|_2^2$ and the functionals 
$$
d(\mathbf{can}(\WF(f)), \WF((g)), \text{ or } d(\WF(f), \mathbf{can}^{-1}(\WF(g)),
$$
where $d$ is a suitable distance function between wavefront sets. 

The method that we propose below is able to compute $\WF(g)$ or a digitized version thereof very efficiently.

\section{Computing the digital wavefront set with shearlets and deep learning}\label{sec:Algorithm}

We propose an algorithm that replaces the heuristic approach of the shearlet-based edge detection and classification algorithm of \cite{yi2009shearlet} by a data-driven approach.
Concretely, instead of hand-crafted heuristics, we train a deep neural network using a variety of training data, adapted to the classification procedure at hand. The neural network takes as input the shearlet coefficients of an image and produces a set of point-direction pairs that are classified as elements of the wavefront set. We will present the construction of the classifier below and then present the computational realization of our algorithm in Subsection \ref{sec:DeNSEE} at the end of this section.

\subsection{Digital shearlet transform} \label{sec:DigShearTrafo}


The classifier to be constructed below is based on the shearlet transform of a digital image.
Therefore, we need to work with a digitized shearlet transform, defined on a digital domain of pixel images. The digital shearlet transform was introduced in \cite{kutyniok2012digital} and is defined as follows:

Let $M\in \Natural$, $J \subset \Natural$ be finite, $k_j \subset \Natural$ for all $j \in J$ and $K_j \coloneqq [-k_j, \dots, 0, \dots, k_j]$.
We pick $2 \sum_{j \in J} K_J +1$ matrices in $\Real^{M \times M}$.
We denote these matrices by $\phi^{dig}$ and $\psi_{j,k, \iota}^{dig} \text{ for } j \in J, k \in K_j, \iota \in \{ -1,1\}$.
To make the connection to the classical shearlet transform, we can think of $\psi_{j,k, \iota}^{dig}$ as a digitized version of $\psi_{2^{-j}, 2^{-j/2}k, 0, \iota}$ and of $\phi^{dig}$ as a digitized version of a low frequency filter.
A concrete construction of the matrices $\phi^{dig}$ and $\psi_{j,k, \iota}^{dig}$ can be found in \cite{kutyniok2012digital}. Then, we define the digital shearlet transform of an image $I \in \Real^{M \times M}$ by
\[
\mathrm{DSH}(I)(j,k,m,\iota) \coloneqq \begin{cases}
  \bigl\langle  
    I, T_m \psi_{j,k, \iota}^{dig}
  \bigr\rangle 
  & \text{ if $\iota \in \{-1,1\}$,} \\[0.5em]
  \bigl\langle  
    I, T_m \phi^{dig} 
  \bigr\rangle 
  & \text{ if $\iota = 0$,}
\end{cases}
\]
where $j \in J, k \in K_j$, $m \in \{1, \dots, M\}^2$, and $T_m: \Real^{M \times M} \to \Real^{M \times M}$ circularly shifts the entries of the elements of a matrix by $m$.

Thus the digital shearlet transform of an image $I \in \Real^{M \times M}$ is a stack of $2 \sum_{j \in J} (K_j-1) +1$ matrices of dimension $M \times M$.
In all our numerical experiments, we fixed $J = 4$ and $K_j = 2^{\lceil j/2 +1\rceil}+1$ so $2 \sum_{j \in J} (K_j-1) +1 = 49$.
The computation of the digital shearlet transform is performed by using the Julia implementation of ShearLab \cite{Kutyniok:2016:SFD:2888419.2740960} (\url{www.shearlab.org/software}).

\subsection{Network architecture} \label{sec:NetArch}

We train a neural network which, given a patch of the shearlet coefficients of a function, produces a prediction of which 
directions belong to the wavefront set of the function at the position associated with this patch. These patches are of size $21 \times 21 \times 49$.

The network architecture consists of four convolutional layers, with $2\times 2$ max pooling, ReLU activation and batch normalization, followed by a fully-connected layer with 1024 neurons, softmax activation function and a one-dimensional output. The network architecture is depicted in Figure~\ref{fig:NetworkArch}. 
We chose this architecture since it performed well in a series of tests while being of moderate size. Here we focused on networks with only a few layers because we expect that the shearlet transform already acts as the correct feature extractor of the problem. 
Therefore, the classifier does not need to learn the correct data representation. Nonetheless, it is conceivable that a deeper and larger neural network architecture could potentially lead to improvements for the classification results below.

We pick 180 directions $(\theta_i)_{i = 1}^{180}$. For each $\theta_i$, we then train a network $\Phi_i$ with the described architecture by passing patches of shearlet coefficients of images $I \in \Real^{M \times M}$ of the form
\begin{align}\label{eq:ShearletBatches}
(\mathrm{DSH}(I)(j,k,m,\iota))_{j \in J, k \in K_j, \iota \in \{-1,0,1\}, m \in [m^*_1-10, m^*_1+10] \times [m^*_2-10, m^*_2+10]},
\end{align}
where $m^* \in \{11, \dots, M-10\}^2$, to the network. The associated label to a batch of \eqref{eq:ShearletBatches} is $1$ if $I$ has a singularity with direction $\theta_i$ at $m^*$ and $0$ else.
In total, this procedure yields $180$ digital classifiers. We train one more network with the same data, but the label is $1$ if $I$ has no singularity at $m^*$ and $0$ else. This additional classifier is used in test cases where all competing algorithms only perform edge detection and not edge-orientation detection.

The final classifier is constructed by putting all of these 181 networks in parallel, producing one large network with 181 outputs. 
For every $21 \times 21 \times 49$ patch of shearlet coefficients, this classifier generates a vector of length 181 indicating if the underlying function is smooth at the center point of the patch and listing all directions of edges present at the center point.

\begin{figure}[h!]
\centering
\begin{tikzpicture}[x={(1,0)},y={(0,1)},z={({cos(60)},{sin(60)})},
font=\sffamily\small,scale=2]
%

\foreach \X [count=\Y] in {1.5,1.5,1.3,1.1}
{\draw pic (box1-\Y) at (\Y*0.75,-\X/2,0) {fake box=white!70!gray with dimensions 0.35 and {2*\X} and 1*\X};
\node[draw,single arrow, black,fill=black!50] at (\Y*0.75+0.45,0.2,0) {\scriptsize{ReLU}};
}

\foreach \X/\Col in {6.3/red,6.5/red,6.7/green,6.9/green, 7.1/blue, 7.3/blue}
{\draw[canvas is yz plane at x = \X*0.7, transform shape, draw = red, fill =
\Col!50!white, opacity = 0.5] (0,0.5) rectangle (2,-1.5);}

\draw[gray!60,thick] (6.1*.7,-0.1,-1.6) coordinate (1-1) -- (6.1*.7,-0.1,0.6) coordinate (1-2) -- (6.1*.7,2.,0.6) coordinate (1-3) -- (6.1*.7,2.1,-1.6) coordinate (1-4) -- cycle;
\draw[gray!60,thick] (7.5*.7,-0.1,-1.6) coordinate (2-1) -- (7.5*.7,-0.1,0.6) coordinate (2-2) -- (7.5*.7,2.,0.6) coordinate (2-3) -- (7.5*.7,2.1,-1.6) coordinate (2-4) -- cycle;
\foreach \X in {4,1,3}
{\draw[gray!60,thick] (1-\X) -- (2-\X);}

\node[draw,single arrow, orange,fill=orange!30] at (7.7*0.7,0.5,0) {Flatten};

\node[circle,draw,blue,fill=blue!30] (A1) at (8.6*0.7,1,0) {~~~};
\node[circle,draw,red,fill=red!30,below=4pt of A1] (A2) {~~~};
\node[circle,draw,green,fill=green!30,below=18pt of A2] (A3) {~~~};
\draw[circle dotted, line width=2pt,shorten <=3pt] (A2) -- (A3);

\node[circle,draw,fill=gray!60] (B1) at (10*0.7,1,0) {~~~};
\node[circle,draw,fill=gray!60,below=4pt of B1] (B2) {~~~};
\node[circle,draw,fill=gray!60,below=18pt of B2] (B3) {~~~};
\draw[circle dotted, line width=2pt,shorten <=3pt] (B2) -- (B3);

\node[circle,draw,fill=gray!60] (C1) at (11*0.7,0.7,0) {~~~};
\node[circle,draw,gray,fill=gray!60,below=7pt of C1] (C2) {~~~};

\begin{scope}[on background layer]
\node[orange,thick,rounded corners,fill=orange!30,fit=(A1) (A3)]{};
\node[gray,thick,rounded corners,fill=gray!10,fit=(B1) (B3)]{};
\node[gray,thick,rounded corners,fill=gray!10,fit=(C1) (C2)]{};
\end{scope}

\foreach \X in {1,2,3}
{\draw[-latex] (A\X) -- (B1);
 \draw[-latex] (A\X) -- (B2);
 \draw[-latex] (A\X) -- (B3);
 \draw[-latex] (B\X) -- (C1);}
\end{tikzpicture}
\caption{Illustration of the network architecture forming the foundation of the classifier. This network consists of four convolutional layers and one fully-connected layer. The colored block in the middle represents a stack of the output of the last convolutional layer. The colors correspond to the different channels.}
\label{fig:NetworkArch}
\end{figure}
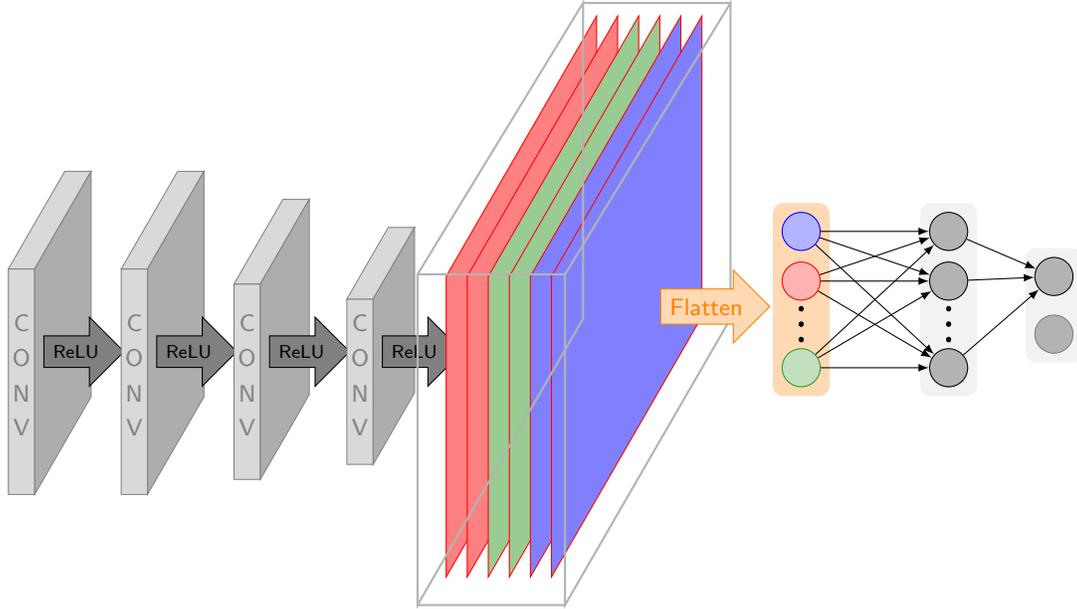

\subsection{Network training}\label{sec:Training}

We train the network as described in Subsection~\ref{sec:NetArch} using stochastic gradient descent to minimize the cross-entropy over a variety of training sets. We use five different data sets to train our classifier and test our algorithm:
\begin{enumerate}
\item[1)] The first data set consists of patches of the shearlet transform of images made of random sums of ellipses and parallelograms of different contrasts, sizes, and orientations.
\item[2)] The second data set contains random sums of ellipses and parallelograms convolved with a kernel to generate functions with a higher-order wavefront set.
\item[3)] The third data set is the BSDS500 (Berkeley Segmentation data set) provided by the Computer Vision Group of UC Berkeley. It comprises 503 natural images of different types. 
\item[4)] The fourth data set is the Semantic Boundaries data set (SBD) with 11355 natural images, again provided by the Computer Vision Group of UC Berkeley.
\item[5)] The fifth data set is constructed by taking Radon transforms of phantoms made of ellipses. The associated wavefront sets are computed analytically through the canonical relations of Definition \ref{def:CanonicalDigital}.
\end{enumerate}

\noindent We depict examples of functions from each of the data sets in Figures~\ref{fig:5}, \ref{fig:6}, \ref{fig:7}, \ref{fig:8}, and \ref{fig:9}.

To make these data sets suitable for our purposes, we need to equip each image of the data sets with an associated set of labels indicating the associated wavefront set or the set of edges. 
For the first two and the last data set, standard theoretical results on the wavefront sets of characteristic functions allow us to compute the associated wavefront sets analytically.
The segmentation and semantic boundaries data sets, on the other hand, are natural images where such an approach is not possible. These data sets are used to assess the quality of segmentation and contour detection applications, see \cite{hariharan2011semantic} and \cite{arbelaez2011contour}. Therefore, every image in these data sets was annotated and has a set of ground truth edges. However, we should point out that this annotated ground truth does not contain all edges of the images, but only those between semantically different parts of the images. We depict the annotated edges in Figures~\ref{fig:5}, \ref{fig:6}, \ref{fig:7}, \ref{fig:8}, and \ref{fig:9}.

In the following subsections, we describe the computation of the associated wavefront sets in detail.

\subsubsection{Ellipses and parallelograms}

The wavefront sets of characteristic functions of ellipses and parallelograms can be identified by \eqref{eq:wavefrontSet} and the fact that if $x$ is a vertex of a parallelogram $P$ then $\{x\} \times \mathbb{S}^1\subset \WF(\chi_P)$.
For sums of these functions, we have, by basic properties of the Fourier transform that
\[
    \WF(\chi_{P_1} + \chi_{P_2}) \subset \WF(\chi_{P_1}) \cup \WF(\chi_{P_2}).
\]
Note that in this relation we do not have equality in general. Indeed, if  $\WF(\chi_{P_1}) \cap \WF(\chi_{P_2})\neq \emptyset$ then cancellations can occur. We shall neglect this technicality as the probability of cancellations is sufficiently small and assume that the wavefront set of characteristic functions as described above is the union of the respective wavefront sets.

We build this data set by randomly choosing a number of parallelograms and ellipses with random positions and computing the associated ground truth of the wavefront set as described above.

\subsubsection{Higher-order wavefront data set}

The ellipses/parallelograms data set contains images with jump singularities only.
To test our method on functions with higher-order singularities, such as ramp singularities, we computed the convolution of the elements of the ellipses/parallelograms data set with a filter $h$ the Fourier transform of which is given by:
$$
\hat{h}(\xi) = \frac{1}{1+|\xi|}, \text{ for } \xi \in \Real^2.
$$
It is not hard to see that $P : f \mapsto h*f$ is an elliptic pseudo-differential operator and hence $\WF(h * g) = \WF(g)$ for all $g \in L^2(\Real^2)$, see \cite[Chapter 8 G]{folland1995introduction} for details. Thus, the convolutions of the elements of the ellipses/parallelograms data set with $h$ have the same wavefront set as the associated ellipses or parallelograms, but of a higher order.

\subsubsection{Segmentation and semantic boundaries data sets}
In the BSDS500 and the SBD data sets, the ground truth of the edges is given in form of binary images with 0's at positions where the image is smooth and 1's at locations associated to edges. This annotated edge set is depicted in Figure~\ref{fig:7}.

To compute the orientation of the edges, we used a five-point stencil derivative on the edges to approximate the normal vectors. To detect corners and assign the appropriate orientations we used the Harris corner detector \cite{harris1988combined}.
From these images, we produce patches for the training of the network classifier. However, due to the fact that the annotated image does not contain all edges we only use patches that are close to these edges for training, validation and testing. 

\subsubsection{{Tomographic} data}

We generate {tomographic data simulated from 2D phantoms similar to the commonly used Shepp-Logan phantom. Each phantom is represented by an $N\times N$ matrix, where} $N\times N$ is the resolution of the images. As in the case of the ellipses/prallelograms data set, these phantoms have an analytic wavefront set given by the union of the wavefront sets of all the ellipses. These are stored as binary matrices of dimensions $N \times N \times 180$ where a $1$ at position $(i,j,\theta)$ indicates that the direction with angle $\theta$ is in the wavefront set at point $(i,j)$. Next, we {simulate tomographic data (sinograms) using the Radon transform implementation in the ODL (\url{http://github.com/odlgroup/odl}) python library for inverse problems.} 

To label {each sinogram} with the correct wavefront set, we map the known wavefront sets of the {corresponding phantom using} a digitized version of the {microlocal canonical relation for the Radon transform} \eqref{eq:CanonicalMap} defined as follows:

\begin{definition}\label{def:CanonicalDigital}
Let $N \in \mathbb N$, we define 
\begin{align*}
    \mathbf{can}^N: \Real^{N \times N \times 180} &\to \Real^{N \times 180 \times 180}\\
    \mathbf{can}^N(X) &= Y,  
\end{align*}
where $Y_{s,\varphi,\lambda} = 1$ if and only if 
\begin{align*}
s &\in \left\lfloor [i_x, j_x] \cdot \left[\cos\left( \frac{\theta_x \cdot \pi}{180}\right), \sin\left(\frac{\theta_x\cdot \pi}{180}\right)\right]^T \right\rfloor + N \cdot \Integer, \qquad
\varphi  = \theta_x \text{ and }\\
\lambda &\in  \left\lfloor \frac{180}{\pi} \cdot \arctan\left(\left[\frac{i_x}{N}, \frac{j_x}{N}\right] \cdot \left[\sin\left( \frac{\theta_x\cdot \pi}{180}\right), -\cos\left(\frac{\theta_x\cdot \pi}{180}\right)\right]^T\right) \right\rfloor + 180 \cdot \Integer,
\end{align*}
for $(i_x,j_x,\theta_x)$ such that $X_{i_x,j_x,\theta_x} = 1$. We call $\mathbf{can}^N$ the \emph{digital canonical map}.

Moreover, we define an \emph{inverse digital canonical map} by 
\begin{align*}
    \mathbf{ican}^N: \Real^{N \times 180 \times 180} &\to \Real^{N \times N \times 180}, \qquad \mathbf{ican}^N:(Y) = X, 
\end{align*}
where $X_{i,j,\theta} = 1$, if and only if 
\begin{align*}
s_y &\in \left\lfloor [i, j] \cdot \left[\cos\left( \frac{\theta\cdot \pi}{180}\right), \sin\left(\frac{\theta\cdot \pi}{180}\right)\right]^T \right\rfloor + N \cdot \Integer, \qquad
\theta  = \varphi_y \text{ and }\\
\lambda_y &\in  \left\lceil \frac{180}{\pi} \cdot \arctan\left(\left[\frac{i}{N}, \frac{j}{N}\right] \cdot  \left[\sin\left( \frac{\theta\cdot \pi}{180}\right), -\cos\left(\frac{\theta \cdot \pi}{180}\right)\right]^T\right) \right\rceil + 180 \cdot \Integer,
\end{align*}
for a $(s_y,\varphi_y,\lambda_y)$ such that $Y_{s_y,\varphi_y,\lambda_y} = 1$. 
\end{definition}
We take Definition \ref{def:CanonicalDigital} as an Ansatz for the definition of the digital wavefront set of the digital Radon transform of a function of which we know the wavefront set. We depict one example of the elements of this data set in Figure \ref{fig:9}.

\subsection{DeNSE: Deep Network Shearlet Edge Extractor}\label{sec:DeNSEE}

We present our algorithm extracting the wavefront set of a digital image below. For $M \in \Natural$, and a digital image $I \in \Real^{M \times M}$, this algorithm produces, for every $m^* \in [11, M-10 ]^2$ a prediction of the wavefront set of $I$ at $m^*$. The algorithm proceeds along the following three steps: 
\begin{description}
    \item[\textbf{Step 1}:] Train the network classifier as of Subsection \ref{sec:NetArch} on a set of labeled training data.
    \item[\textbf{Step 2}:] For a given test image $I \in \Real^{M \times M}$, compute the digital shearlet transform of $I$ with $49$ shearlet generators: $(\mathrm{DSH}(I)(j,k,m,\iota))_{j \in J, k \in K_j, \iota \in \{-1,0,1\}, m \in [1, M]^2}$.    
    \item[\textbf{Step 3}:]
    For every $m^* = (m^*_1, m^*_2) \in [11, M-10]^2$, pass the patch
\begin{align} \label{eq:digitalShearletCoeffs}
\bigl(\mathrm{DSH}(I)(j,k,m,\iota) \bigr)_{j \in J, k \in K_j, \iota \in \{-1,0,1\}, m \in [m^*_1-10, m^*_1+10]\times [m^*_2-10, m^*_2+10]}
\end{align}
to the classifier of Step 1. If the classifier predicts that an edge with direction $\theta$ is present, then classify $(m^*, \theta)$ as an element of the wavefront set of $I$. 
 \end{description}
Henceforth, we refer to the above algorithm as Deep Network Shearlet Edge Extractor (DeNSE).

\section{Numerical results} \label{sec:numRes}

We implemented the training as described in the previous section using the GPU version of Tensorflow. To evaluate the classification quality, we use two quality measures, a mini-batch test average taken over all mini-batches and the so-called MF-score. The MF-score is computed as the mean of the F-score defined as
\[ F \coloneqq \frac{2 P R}{R+P}, \]
where $P$ is the \emph{precision}, i.e., the number of true positives divided by the sum of true and false positives, and $R$ is the \emph{recall}, i.e., the number of true positives divided by the sum of true positives and false negatives, \cite{sasaki2007truth}. The MF-score is often used for evaluating classification performance when the distribution of classes is uneven. This is, for example, the case in edge detection, since there usually are significantly fewer edge points than smooth points in an image. Moreover, these performance measures enable us to compare with the state-of-the-art \cite{yu2018simultaneous} on the respective data sets. In addition, the required code to reproduce the results is publicly available in {\url{http://www.shearlab.org/applications}}.

The implementations of the other methods that were used to compare the performance of DeNSE were taken from the publicly available github repositories provided by the authors of the methods (all available in the topic: \url{http://github.com/topics/edge-detection}), with exception of the Canny, Sobel algorithms that were taken from the python library OpenCV (\url{http://opencv.org/}). The data-driven methods were trained using the hyper-parameters proposed by the authors for the given {data set} (Berkeley segmentation set and SBD) without further tuning.

\subsection{Results for ellipses/parallelograms}

We train each of the 181 subnetworks as of Subsection~\ref{sec:NetArch} using 10,000 images as training data, 1,000 images as validation data, and 2,000 images as test data. For each direction $\theta_i$ we trained the associated subnetworks using a mini-batch procedure with 86 examples per batch and 3,000 training steps for each. We obtained an average test accuracy of 96.2\% (taking the average over all 181 classifiers) and an MF-score of 97.1\%. We also notice that the test accuracy of the individual classifiers was higher when classifying angles aligned to the discrete orientations of the underlying shearlet system.

We compared our method on the data set of the shearlet coefficients patches to other classifiers commonly used in machine learning namely: Logistic regression, Decision trees, $K$-nearest neighbors, Linear SVM, and Random forest, using their implementations in the python library {scikit-learn (\url{http://scikit-learn.org/})}. We report the performances of these classifiers in Table~\ref{table:Ellipses-accuracy}.

\begin{table}[htb!]
\centering
\begin{tabular}{l r r r r}
\textbf{Method} & \textbf{Test accuracy}  & \textbf{MF-score}  \\
\hline
Logistic regression & 45.7 &  48.9\\
\hline

Decision trees & 75.2 & 75.8 \\
\hline

Linear SVM & 46.5 & 50.3 \\
\hline

K-nearest neighbors & 72.7 & 73.2 \\
\hline

Random forest & 86.0 & 86.7 \\
\hline

{DeNSE} & \emph{96.2} & \emph{97.1} \\
\hline
\end{tabular}
\caption{Ellipses/parallelograms data set performance metrics in percentage.}
\label{table:Ellipses-accuracy}
\end{table}

By construction, the last of the 181 subnetworks corresponds to an edge-detector, where the achieved average test accuracy was 97.5\%, and the MF-score was 97.9\%, the performance benchmarks with other classical edge classifiers can be found in Table~\ref{table:Ellipses-edge}. Figure~\ref{fig:5} shows the results on an example of the ellipses/parallelograms data set.

We depict the classification for one instance of the test set of the parallelograms/ellipses data set in Figure \ref{table:Ellipses-edge} and compare the results with the classification by the heuristic approach by Yi-Labate-Easley-Krim \cite{yi2009shearlet}. We observe that our method performs significantly better in low contrast regions. Moreover, our algorithm appears to be more precise when differentiating between corners and edges. Here, we classify a point as a corner point if the classifiers predict at least two different orientations that differ by more than 10 degrees.
In Figures~\ref{fig:5}, ~\ref{fig:6},~\ref{fig:7} and~\ref{fig:8}, we indicate corners by white dots.

\begin{table}[h!]
\centering
\begin{tabular}{l r r r r}
\textbf{Method} & \textbf{MF-score}  \\
\hline
Canny \cite{canny1986computational} &  49.1\\
\hline

Sobel \cite{sobel2014edge} & 40.0 \\
\hline

BEL \cite{bel2006edge} & 63.3 \\
\hline

Yi-Labate-Easley-Krim \cite{yi2009shearlet} & 70.3 \\
\hline

CoShREM \cite{rafael2015coshrem} & 90.6 \\
\hline

{DeNSE} & \emph{97.5}  \\
\hline
\end{tabular}
\caption{Edge detection performances of edge detection algorithms on the Ellipses/parallelograms data set. The MF-Score is in percentage.}
\label{table:Ellipses-edge}
\end{table}

\begin{figure}[htb!]
\centering
\includegraphics[width = 0.38\textwidth]{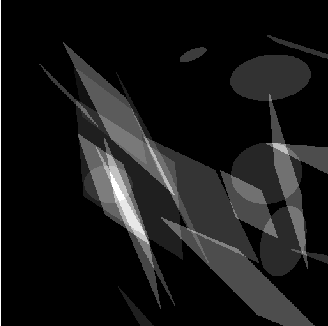}
\includegraphics[width = 0.38\textwidth]{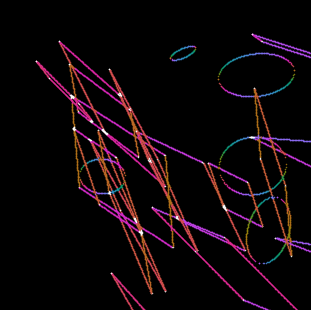} 
\\
\includegraphics[width = 0.38\textwidth]{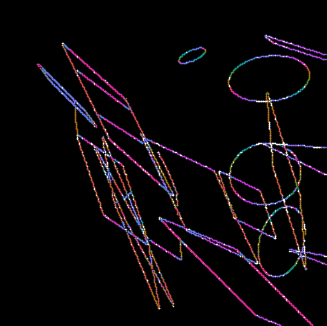}
\includegraphics[width = 0.38\textwidth]{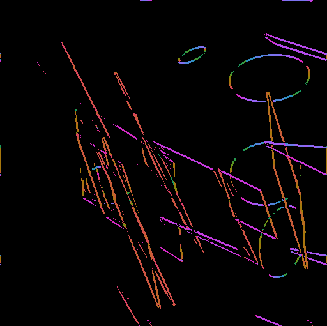}
\\ 
\includegraphics[width = 0.38\textwidth]{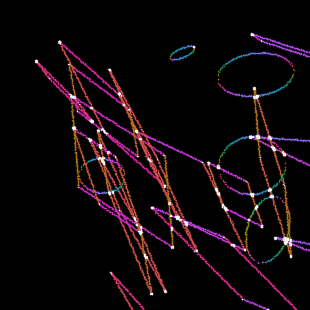}
\includegraphics[width = 0.38\textwidth]{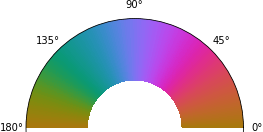}
\caption{Computed edges and orientations of an example of the ellipses/parallelograms data set. Top-left: Input image. Top-right: Orientations, human annotation. Middle-left: Orientations predicted by Yi-Labate-Easley-Krim algorithm. Middle-right:  Orientations predicted by CoShREM. Bottom-left: Orientations predicted by DeNSE\@ algorithm. Bottom-right: Color code for normal-directions.}
\label{fig:5}
\end{figure}

\begin{figure}[htb!]
\centering
\includegraphics[width = 0.38\textwidth]{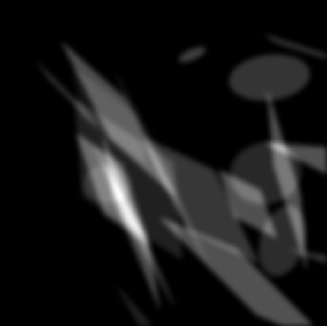}
\includegraphics[width = 0.38\textwidth]{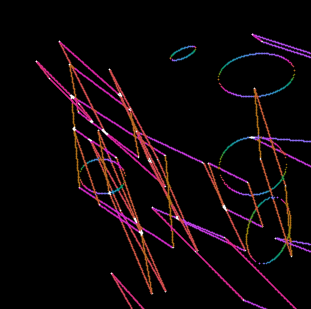}
\\
\includegraphics[width = 0.38\textwidth]{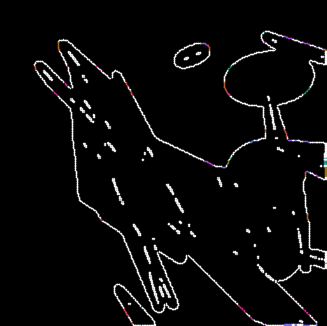}
\includegraphics[width = 0.38\textwidth]{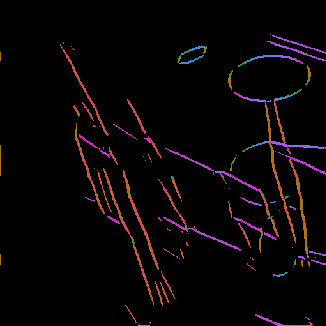}
\\
\includegraphics[width = 0.38\textwidth]{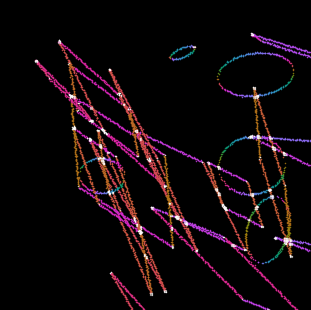}
\includegraphics[width = 0.38\textwidth]{./cmap2.png}

\caption{Computed edges and orientations of an example of the higher-order ellipses/parallelograms data set. Top-left: Input image. Top-right: Orientations, human annotation. Middle-left: Orientations predicted by Yi-Labate-Easley-Krim algorithm. Middle-right:  Orientations predicted by CoShREM. Bottom-left: Orientations predicted by DeNSE\@ algorithm. Bottom-right: Color code for normal-directions.}
\label{fig:6}
\end{figure}

\begin{table}[htb!]
\centering
\begin{tabular}{l r r r r}
\textbf{Method} & \textbf{MF-score} \\
\hline
gPb-owt-ucm \cite{arbelaez2011contour} & 73.7 \\
\hline

gPb \cite{arbelaez2011contour} & 71.5 \\
\hline

Mean Shift \cite{comaniciu2002meanshift} & 64.0 \\
\hline

Normalized Cuts \cite{shi1997cuts} & 64.2 \\
\hline

Felzenszwalb, Huttenlocher \cite{felzenszwalb2004segment}& 61.0 \\
\hline

Canny & 60.3 \\
\hline

CoShREM \cite{rafael2015coshrem} & 75.7 \\
\hline

DeepEdge \cite{bertasius2015deepedge} & 75.3 \\
\hline

{DeNSE} & \emph{95.4}\\
\hline
\end{tabular}
\caption{BSDS500 (Berkeley) data set performance metrics in percentage.}
\label{table:BSD-accuracy}
\end{table}

\begin{figure}[htb!]
\centering
\includegraphics[width = 0.38\textwidth]{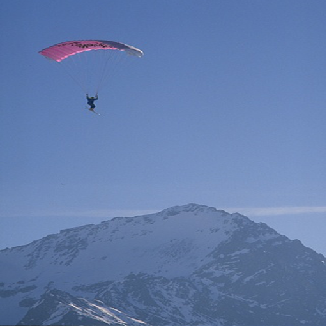}
\includegraphics[width = 0.38\textwidth]{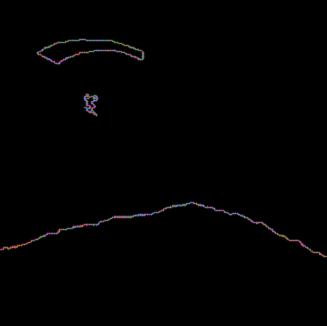}
\\
\includegraphics[width = 0.38\textwidth]{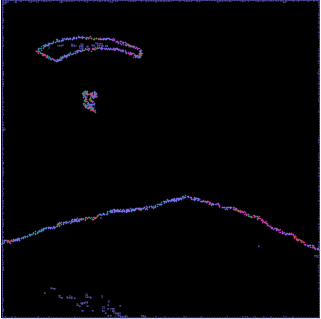}
\includegraphics[width = 0.38\textwidth]{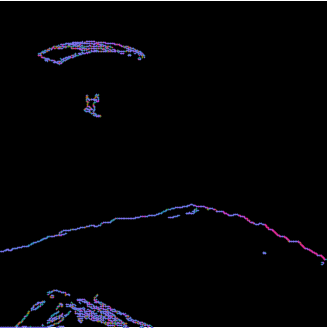}
\\ 
\includegraphics[width = 0.38\textwidth]{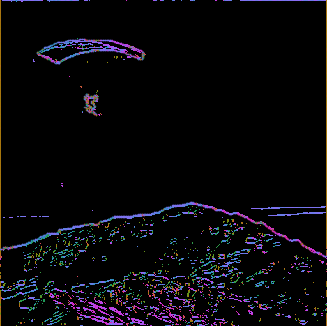}
\includegraphics[width = 0.38\textwidth]{./cmap2.png}
\caption{Computed edges and orientations of an example of the BSDS500 (Berkeley) data set.
Top-left: Input image. Top-right: Orientations, human annotation. Middle-left: Orientations predicted by gPb-owt-ucm. Middle-right:  Orientations predicted by CoShREM. Bottom-left: Orientations predicted by DeNSE\@ algorithm. Bottom-right: Color code for normal-directions.}
\label{fig:7}
\end{figure}

\subsection{Results for higher-order wavefront set data set}
We performed wavefront set detection for the the higher-order wavefront set data set, using the same procedure as in the ellipses/parallelograms classification. In this case, we used 30,000 patches as training data, 3,000 patches as validation data, and 6,000 patches as test data. We trained on 86-sized mini-batches, with 200,000 training steps. We obtained an average test accuracy of 93.4\% and an MF-score of 94.6\%. We are not aware of any algorithms specifically build for higher-order wavefront set detection, which is why we do not provide a comprehensive list of results of alternative algorithms in this case.

For completeness, we added Figure~\ref{fig:7} showing an example of the obtained results. We also add two predictions by the algorithm of Yi-Labate-Easley-Krim \cite{yi2009shearlet} and the method CoShREM  \cite{rafael2015coshrem}. 

It is important to mention that the algorithm of Yi-Labate-Easley-Krim is constructed to detect jump singularities and not higher-order singularities. Hence this algorithm is expected to fail on this data set. Indeed, the performance of the algorithm of achieves only an MF-score of 30.5\%. CoShREM, on the other hand, is built to detect edges and ridges. The performance was significantly better than that of Yi-Labate-Easley-Krim and resulted in an MF-score of 65.4\%.

\subsection{Results for Berkeley segmentation set}

In the Berkeley segmentation data set, the complexity of the images is considerably higher compared to the images from the ellipses/parallelogram data set. Therefore, we use a significantly larger training set to train the associated classifier. For the classification of each angle, we used 30,000 patches as training data (around 600 patches per image), 3,000 patches as validation data, and 6,000 patches as test data. As in the case of the ellipses/parallelograms, we train using a mini-batch procedure, with 86 examples per batch, but in this case, using 30,000 training steps for each. We obtained an average test accuracy of 93.1\% and MF-score of 95.4\%, which is lower than the one obtained in the ellipses/parallelogram. This is likely due to the higher complexity of the patches.
One advantage of this and the SBD data set is the existence of several benchmarks including state-of-the-art deep learning based algorithms.

We compared our method using the available benchmarks on this data set provided by the UC Berkeley Computer Vision Group, we refer to \cite{arbelaez2011contour} for a more detailed explanation of these methods. In \cite{arbelaez2011contour}, just the MF-score of the competing algorithms was reported. We give the results in Table~\ref{table:BSD-accuracy}.

We present one example of the results obtained on the BSDS500 data set in Figure~\ref{fig:7}.

\subsection{Results for semantic boundary set (SBD)}

The SBD data set contains significantly more images than the BSDS500 which, as we will observe below, improves the overall classification performance slightly. In this case, we used 100,000 patches as training data, 10,000 patches as validation data, and 20,000 patches as test data. We train on 86-sized mini-batches, with 100,000 training steps. We obtained average test accuracy of 95.3\% and MF-score of 96.8\%.

This data set has recently been widely used for image segmentation tasks, in particular, it was used on the two deep learning based image segmentation frameworks proposed by Z. Yu \textit{et al.}, namely the SEAL (Simultaneous Edge Alignment and Learning) \cite{yu2018simultaneous} and the CASENet (Category-Aware Semantic Edge Detection Network) \cite{yu2017casenet}. We also compared them with the deep learning image boundary detector and classifier proposed (OBDC) by J. Y. Koh \textit{et al.} \cite{koh2017object}. The results can be found on Table~\ref{table:SBD-accuracy}.

Figure~\ref{fig:8} shows the results obtained by DeNSE on an example image of the SBD data set, as in the case of the BSDS500 data set, the obtained result admits more edges than the ground truth due to the batch-based approach. Nonetheless, the method outperforms even the specialized algorithms for segmentation over the given data sets.

\begin{table}[htb!]
\centering
\begin{tabular}{l r r r r}
\textbf{Method} & \textbf{MF-score} \\
\hline
OBDC & 62.5 \\
\hline

CASENet & 71.8 \\
\hline

CASENet-S & 75.8 \\
\hline

CASENet-C & 80.4 \\
\hline

CoShREM & 69.7\\
\hline

SEAL & 81.1 \\
\hline

\emph{DeNSE} & \emph{96.8}\\
\hline
\end{tabular}
\caption{Performance on the SBD data set. All values are in percentage.}
\label{table:SBD-accuracy}
\end{table}

\begin{figure}[htb!]
\centering
\includegraphics[width = 0.38\textwidth]{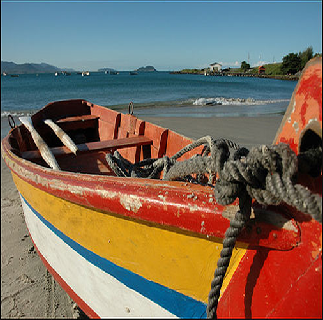}
\includegraphics[width = 0.38\textwidth]{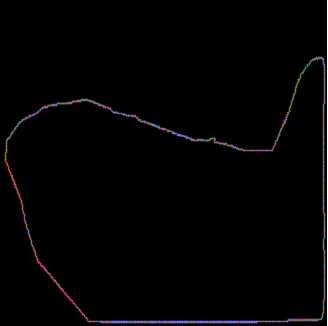}
\\
\includegraphics[width = 0.38\textwidth]{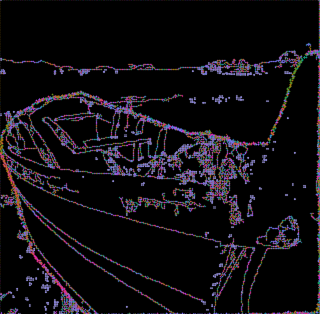}
\includegraphics[width = 0.38\textwidth]{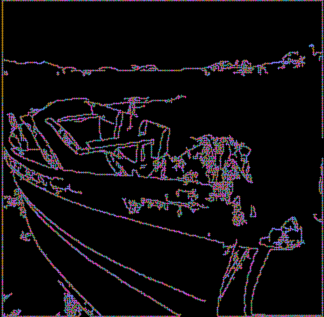}
\\ 
\includegraphics[width = 0.38\textwidth]{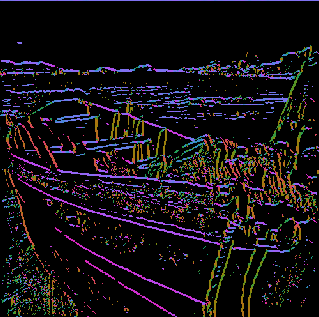}
\includegraphics[width = 0.38\textwidth]{./cmap2.png}
\caption{Computed edges and orientations of an example of the SBD data set.Top-left: Input image. Top-right: Orientations, human annotation. Middle-left: Orientations predicted by  SEAL. Middle-right: Orientations predicted by CoShREM. Bottom-left: Orientations predicted by DeNSE\@ algorithm. Bottom-right: Color code for normal-directions.}
\label{fig:8}
\end{figure}

\subsection{Results for tomographic data} \label{sec:SinogramResults}

{Learning a data-driven wavefront set extractor for tomographic data uses 10\,000 images with associated wavefront sets as supervised training data, 1\,000 images as validation data, and 2\,000 images as test data.} 
In this case, we obtained an average test accuracy of $94.5\%$ and $95.7\%$ of MF-score. This score is comparable to the performances of DeNSE on the other image classes. In this case, we could not perform any comparison with alternative methods, because we are not aware of any competing algorithm for the detection of wavefront sets of sinograms.

Extracting the wavefront set of the data and then applying canonical relations has the computational advantage that we do not need to solve the inverse problem. On top of this, it seems to also give improved wavefront set detection of the original signal over strategies that first solve the inverse problem and then detect the wavefront set. To substantiate this claim, we perform the following test. Using three standard inversion schemes, filtered backprojection, Tikhonov-regularized and total variation regularized inversion, we first compute an approximate inverse from a low dose CT image and then compute the associated wavefront set using DeNSE. We then note the average mean square error to the true wavefront set of the data point. We compare this with the error resulting from computing the wavefront set of the sinogram and transforming this to the wavefront set of the image through the canonical relations. We depict the wavefront sets extracted in this way in Figure \ref{fig:10}.

\begin{table}[htb!]
\centering
\begin{tabular}{l r r r r}
\textbf{Inversion technique} & \textbf{Mean square error}\\
\hline
Tikhonov & 443.0\\
\hline

Total variation & 380.9\\
\hline

Filtered backprojection & 504.3 \\
\hline
Canonical relations & 168.1 \\
\hline
\end{tabular}
\caption{Error of wavefront set estimation by different inversion techniques.}
\label{table:TV-results}
\end{table}

The results of Table \ref{table:TV-results} clearly demonstrate the advantage of first extracting the wavefront set of the sinogram and then applying the canonical relations over any first-invert-then-extract strategy.

\begin{figure}[htb!]
\centering
\includegraphics[width = 0.49\textwidth]{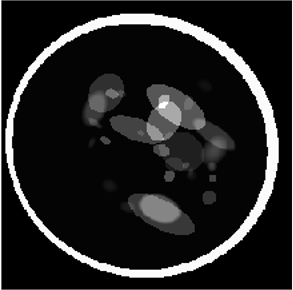}
\includegraphics[width = 0.49\textwidth]{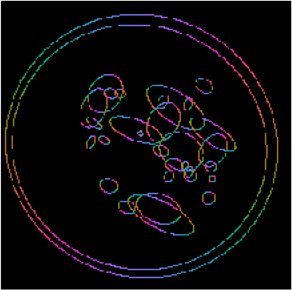}
\\[1em]
\includegraphics[width = 0.49\textwidth]{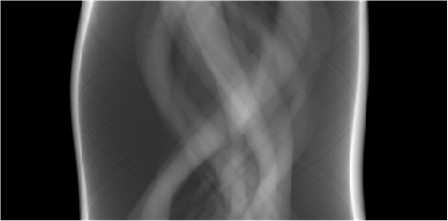} 
\includegraphics[width = 0.49\textwidth]{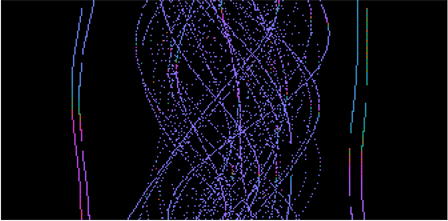}
\\[1em]
\includegraphics[width = 0.38\textwidth]{./cmap2.png}
\caption{Top-left: Phantom made from ellipses. Top-right: Associated wavefront set extracted by DeNSE. Middle-left: Radon transform of the phantom. Middle-right: Associated wavefront set computed through digital canonical relations of Definition \ref{def:CanonicalDigital}. Bottom: Color-code for normal directions.}
\label{fig:9}
\end{figure}

\begin{figure}[htb]
\centering
\includegraphics[width = 0.38\textwidth]{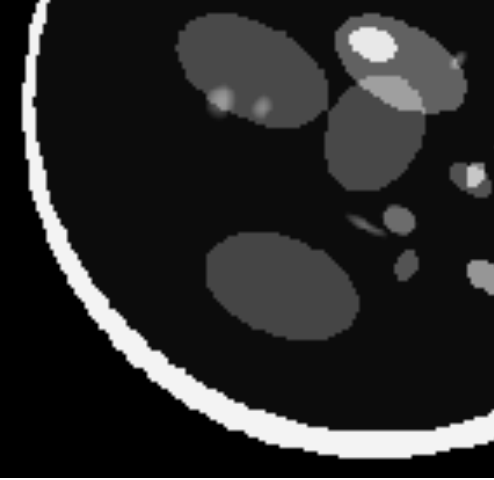}
\includegraphics[width = 0.38\textwidth]{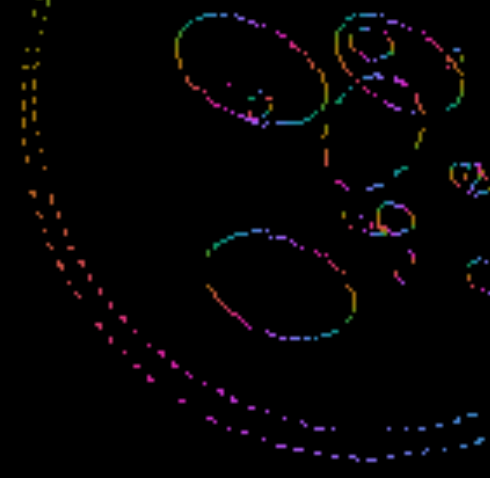}
\\
\includegraphics[width = 0.38\textwidth]{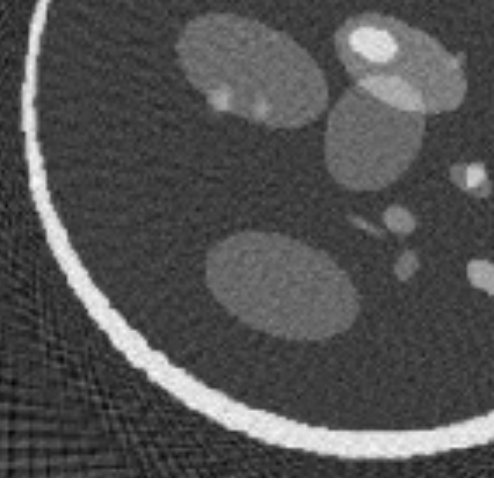}
\includegraphics[width = 0.38\textwidth]{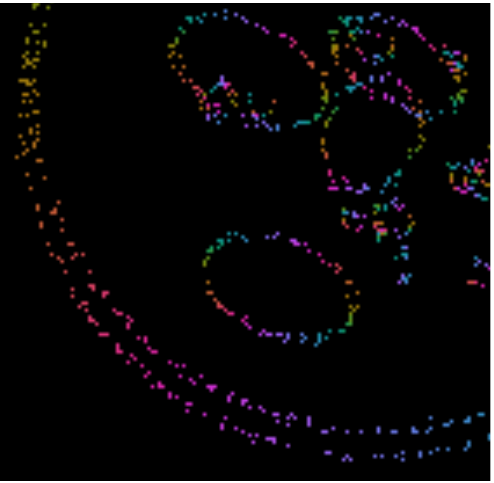}
\\ 
\includegraphics[width = 0.38\textwidth]{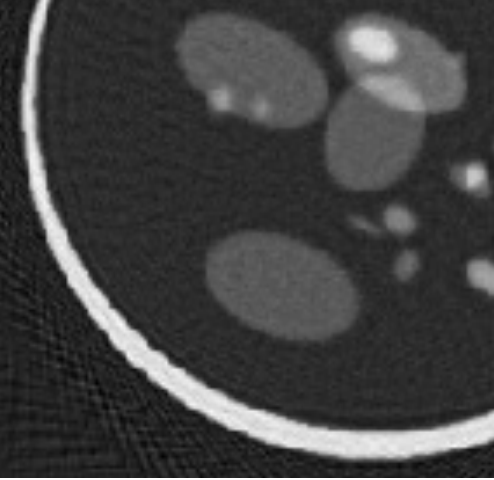}
\includegraphics[width = 0.38\textwidth]{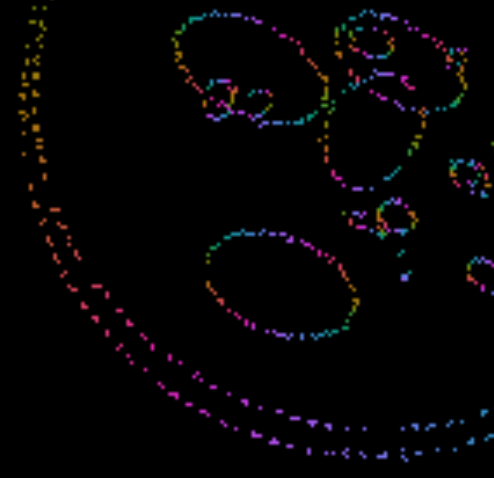}
\caption{Top-left: Phantom made from ellipses. Top-right: Wavefront set of the phantom data computed by the inverse canonical relation on the low dose sinogram wavefront set extracted by DeNSE. 
Middle-left: Filtered back projection reconstruction. Middle-right: Wavefront set of the filtered backprojection reconstruction extracted by DeNSE. 
Bottom-left: Tikhonov reconstruction. Bottom-right: Wavefront set of the Tikhonov reconstruction extracted by DeNSE.}\label{fig:10}
\end{figure}
\subsection{Conclusion}

We observe that in all performed tests, our novel algorithm DeNSE significantly outperforms all competitors. In doing so, DeNSE outperforms not only traditional methods, but also other, deep learning-based algorithms. Comparing the complexity of the involved neural networks reveals that the classifier of DeNSE uses a comparably small neural network.

One natural explanation for the jump in performance already with simple networks is the fact that the shearlet representation transforms the data in a much more convenient form for training and evaluation purposes at least from the point of view of wavefront set extraction.

Moreover, we observed strong performance of DeNSE for extraction of non-edge-like wavefront sets. In particular, the successful detection of the wavefront set of sinograms suggest a possibility to incorporate this information into regularizing schemes of inverse problems or for analysis purposes. 

We observed that through the canonical relations this approach allows the detection of the wavefront set of a phantom without solving the inverse problem. This is particularly interesting since this approach does not only circumvent an inversion step, but also improves the performance over algorithms that first invert the Radon transform and then extract the wavefront set.


\section*{Acknowledgements}

H.A.-L. is supported by the Berlin Mathematical School. G. K. acknowledges partial support by the Bundesministerium f\"ur Bildung und Forschung (BMBF) through the Berliner Zentrum for Machine Learning (BZML), Project AP4, by the
Deutsche Forschungsgemeinschaft (DFG) through grants CRC 1114 "Scaling Cascades in Complex Systems",
Project B07, CRC/TR 109 "Discretization in Geometry and Dynamics', Projects C02 and C03,
RTG DAEDALUS (RTG 2433), Projects P1 and P3, RTG BIOQIC (RTG 2260), Projects P4 and P9,
and SPP 1798 "Compressed Sensing in Information Processing", Coordination Project and Project Massive
MIMO-I/II, by the Berlin Mathematics Research Center MATH+ , Projects EF1-1 and EF1-4, and
by the Einstein Foundation Berlin. The work of O.\"O. was supported by the Swedish Foundation of Strategic Research grant AM13-004. P.P is supported by a DFG Research Fellowship ”Shearlet-based energy functionals for anisotropic phase-field
methods”.

\bibliographystyle{abbrv}
\bibliography{references}
\end{document}